\documentclass[11pt]{article}
\usepackage{graphicx}
\usepackage{amssymb}
\usepackage{amsmath}
\usepackage{amsthm}
\usepackage{mathrsfs}
\usepackage{natbib}
\usepackage{enumerate}
\usepackage{hyperref}
\hypersetup{colorlinks=true,linkcolor=blue, citecolor=blue}
\usepackage{caption}
\usepackage{subcaption}
\usepackage{lmodern}
\usepackage{lineno}
\usepackage{ulem}
\usepackage{soul}
\usepackage{cancel}
\usepackage{float}
\usepackage{color}
\usepackage{multirow}
\usepackage{bm}
\usepackage{enumitem}
\usepackage{appendix}
\usepackage{csquotes}
\usepackage{lipsum}
\usepackage{booktabs}
\usepackage{commath}
\usepackage{mathtools}
\usepackage{authblk}
\usepackage[margin = 2cm]{geometry}

\newcommand{\bL}{\mathbb{L}}
\newcommand{\bS}{\mathbb{S}}
\newcommand{\BB}{\mathbb{B}}

\newcommand{\CC}{\mathbb{C}}

\newcommand{\II}{{\mathbb{I}}}

\newcommand{\NN}{{\mathbb{N}}}
\newcommand{\RR}{{\mathbb{R}}}

\newcommand{\bn}{{\bigskip\noindent}}
\newcommand{\n}{{\noindent}}

\allowdisplaybreaks
\def\d{{\rm d}}

\newtheorem{prop}{Proposition}

\newtheorem{corollary}{Corollary}

\newtheorem{lem}{Lemma}
\newtheorem{thm}{Theorem}
\newtheorem{remark}{Remark}
\newtheorem{assumption}{Assumption}
\definecolor{uared}{rgb}{0.85, 0.0, 0.3}

\newcommand{\blind}{0}

\begin{document}

\def\spacingset#1{\renewcommand{\baselinestretch}%
	{#1}\small\normalsize} \spacingset{1}

\if0\blind
{
	\title{\bf Testing Symmetry for Bivariate Copulas using Bernstein Polynomials}
	
	\author{Guanjie Lyu\\
		Department of Mathematics and Statistics, University of Windsor, Canada\\
		
		M. Belalia\thanks{
			Corresponding author: \textit{Mohamed.Belalia @ uwindsor.ca}} \hspace{.2cm}\\
		Department of Mathematics and Statistics, University of Windsor, Canada}
	
	\maketitle
} \fi

\if1\blind
{
	\bigskip
	\bigskip
	\bigskip
	\begin{center}
		{\LARGE\bf Testing Symmetry for Bivariate Copulas using Bernstein Polynomials}
	\end{center}
	\medskip
} \fi

\bigskip
\noindent\rule{\textwidth}{0.8pt}
\begin{abstract}
	In this work, tests of symmetry for bivariate copulas are introduced and studied using empirical Bernstein copula process. Three statistics are proposed and their asymptotic properties are established. Besides, a multiplier bootstrap Bernstein version is investigated for implementation purpose. The simulation study demonstrated the superior performance of the Bernstein tests compared to tests based on empirical copulas. Furthermore, in real data applications, these tests consistently yielded similar conclusions across a diverse range of scenarios.
\end{abstract}
\noindent%

\noindent\rule{\textwidth}{0.8pt}
\spacingset{1.45} 

\section{Introduction}
\label{sec:introduction}

\n Asymmetric copulas have massive explorations and applications in recent years. They are powerful tools for capturing the asymmetrical dependence structure and have been applied in many fields, for instance,~\cite{Grimaldi2006} for flood frequency analysis,~\cite{Wu2014} in reliability modelling and~\cite{Zhang2018} for ocean data analysis. Meanwhile, more attention is given to testing and identifying the symmetric nature of a copula.~\cite{Genest2012} proposed tests based on empirical copula,~\cite{Bahraoui2018} used empirical copula characteristic function to construct the test, ~\cite{Jaser2020} developed a test by representing copula
as a mixture of two conditional distribution functions, \cite{Beare2020} studied a randomization procedure. For $d$-variate symmetry tests, the work of~\cite{Genest2012} was extended by~\cite{Harder2017}, and~\cite{Bahraoui&Quessy2021} investigated tests based on L\'evy measures.

In this work, tests of symmetry for bivariate copulas based on  empirical Bernstein copula process are proposed. Specifically, consider a random pair $(X, Y)$ with cumulative distribution function $F_{XY}(x, y)=\mathbb{P}(X\le x, Y\le y)$ and continuous margins $F_X$ and $F_Y$. According to~\cite{Sklar1959}, there exists a unique copula function $C$ such that
\begin{equation*}
	F_{XY}(x, y)=C\{F_X(x), F_Y(y)\}.
\end{equation*}	
To detect the symmetry of copula $C$, one would like to test the following hypotheses
\begin{equation}{\label{eq: hypotheses}}
	\begin{cases}
		\mathscr{H}_0: \forall (u,v) \in [0,1]^2 &  C(u, v)=C(v,u), \\
		& \text{versus}\\
		\mathscr{H}_1: \exists (u,v) \in [0,1]^2 & C(u, v)\ne C(v,u) .
	\end{cases}	 
\end{equation}	

\n  The symmetry property of bivariate copulas has intimate relation with the symmetry of corresponding random pairs and was discussed in~\cite{ Nelsen1993, Nelsen2006} and~\cite{Nelsen2007}. It was shown that $X$ and $Y$ are exchangeable, $i.e.$, $F_{XY}(x,y)=F_{XY}(y,x)$ if and only if $F_X=F_Y$, and $C(u, v)=C(v, u)$ for $(u, v)\in [0, 1]^2$. Specifically, for identically distributed margins, the symmetry structure of the copula can determine the exchangeability of the random variables $X$ and $Y$. Identifying the equality of margins is well-developed and can be investigated using Kolmogorov-Smirnov test or Cram\'{e}r-von Mises test. For a positively dependent survival data setting, see~\cite{Fujii1989} and for a high-dimensional data setting, see~\cite{Cousido-Rocha2019}. Together with the verification of equal margins, approaches to detecting the symmetry of the copula function would provide an effective way to decide on the exchangeability of two random variables. As the equality of margins was already discussed vigorously by many others, the contribution of detecting the symmetry of copula becomes more appealing for deciding the exchangeability of two random variables. Moreover, before fitting a specific copula model to the data, identifying the symmetric structure can assist us to choose an appropriate model, for example, Archimedean copulas are symmetric.

Empirical Bernstein copula has drawn attention in recent years due to free boundary bias properties and tractability of implementation. Indeed, one only needs to select the degree of the Bernstein polynomials during deployment. Theoretical properties of this estimator have been well-studied, see~\cite{Sancetta2004}, \cite{Jansen2012}, \cite{Belalia2017b} and~\cite{Segers2017}, among others. However, computational investigations of the empirical processes based on this estimator are barely explored. To fill this gap, we proposed a smooth version of multiplier bootstrap method for empirical Bernstein copula process with promising performance. 

The rest of this paper is structured as follows. In Section~\ref{sec: Tests}, based on the empirical Bernstein copula, an extension of symmetry test statistics in~\cite{Genest2012} is proposed and their asymptotic behaviours are examined. Section~\ref{sec:Multiplier} develops  the  empirical Bernstein copula process multiplier bootstrap and its large sample behaviour. Simulation studies are carried out in Section~\ref{sec:Sim}. Two real data applications are presented in Section~\ref{sec: realdata}. Some concluding remarks are given in Section~\ref{sec:conclusion}. Finally, the proofs are relegated in the \hyperlink{app}{Appendix} and the \textbf{R} code that was used in this article is available on \href{https://github.com/Tokkkk/EBC_SymmetryTest}{GitHub}.

\section{New testing procedures based on empirical Bernstein copula\label{sec: Tests}}

\subsection{Descriptive of the test statistics}

\n Let $(X_1, Y_1), \ldots, (X_n, Y_n)$ be a random sample from a bivariate distribution $F_{XY}$ with continuous margins $F_X$ and $F_Y$. Also, let $C$ be their associated copula. Unfortunately, this function is generally unknown, hence it has to be estimated. The empirical copula introduced by~\cite{Ruschendorf1976} is a natural nonparametric estimator of $ C $ and is given by 
\begin{equation}\label{eq:(1)}
	\widehat{C}_n(u, v)=\frac{1}{n}\sum_{i=1}^n\II\left(\widehat{U}_i\le u, \widehat{V}_i\le v\right),
\end{equation}
where $\widehat{U}_i=n^{-1}\sum_{j=1}^n\II\left(X_j\le X_i\right)$ and $\widehat{V}_i=n^{-1}\sum_{j=1}^n\II\left(Y_j\le Y_i\right)$ are the empirical distributions of the margins. 

The empirical copula is widely used for construction of nonparametric tests in the literature, such as, test of independence, goodness-of-fit test among others. However, it is not a continuous estimator for $C$, which mismatches the continuity of the copula function. To overcome this drawback, smoothed empirical copulas were developed, for example,~\cite{Morettin2010} introduced wavelet-smoothed empirical copula for time series data, ~\cite{Gijbels1990},~\cite{Fermanian2004},~\cite{Cheng2007},~\cite{Omelka2009} considered kernel-smoothed empirical copula and~\cite{Genest2017} proposed empirical checkerboard copula. Here, the empirical Bernstein copula is employed. This choice is motivated by (i) estimation based on Bernstein polynomials is known to be asymptotically bias free at boundary points (see, \cite{Leblanc2012b}, \cite{Jansen2012}, \cite{Belalia2016}) as compared to kernel based methods which suffer from excessive bias at or near to the boundary points. A good discussion about the boundary bias for kernel based methods can be found in~\cite{Cheng2007}. (ii) The empirical Bernstein copula is a polynomial, hence, it has all partial derivatives, which will be of highly important for building our multiplier bootstrap. Besides, the support of bivariate copula is $[0, 1]^2$ which meets the Bernstein polynomials assumption perfectly. 

The empirical Bernstein copula estimator of order $m$ is defined as 
\begin{align}\label{eq:(2)}
	\widehat{C}_{n, m}(u, v)&=\sum_{k=0}^m\sum_{\ell=0}^m\widehat{C}_n\left(k/m, \ell/m\right)P_{m,k}(u)P_{m,\ell}(v)\notag \\
	&=\frac{1}{n}\sum_{i=1}^n\sum_{k=0}^m\sum_{\ell=0}^m\II\left(\widehat{U}_i\le k/m, \widehat{V}_i\le \ell/m\right)P_{m, k}(u)P_{m, \ell}(v),
\end{align}
where $P_{m,k}(u)=\binom{m}{k}u^{k}(1-u)^{m-k}$ is the binomial probability mass function. Note that, $m$ is dependent on $n$, in particular, if the degree of Bernstein polynomial $m$ is equal to the sample size $n$, the estimator~\eqref{eq:(2)} turns out to be the empirical beta copula developed by~\cite{Segers2017}. Resampling procedures with the empirical beta copula are proposed and studied in~\cite{Kiriliouk2021}. A broad class of smooth, possibly data-adaptive nonparametric copula estimators are proposed in~\cite{Kojadinovic2022smooth} and~\cite{Kojadinovic2022stute}, which includes the empirical Bernstein and beta copula.


In the same spirit as the one presented in~\cite{Genest2012},  Kolmogorov-Smirnov and Cram\'{e}r-von Mises type statistics are proposed. Specifically, these statistics are built upon the empirical Bernstein copula and are presented as follows:
\begin{align}\label{eq:Statest}
	R_{n,m}&=\int_{0}^{1}\int_{0}^1\Big\{\widehat{C}_{n, m}(u, v)-\widehat{C}_{n,m}(v, u)\Big\}^2\dif u \dif v,\notag\\
	S_{n,m}&=\int_{0}^1\int_{0}^1\Big\{\widehat{C}_{n, m}(u, v)-\widehat{C}_{n, m}(v, u)\Big\}^2 \dif \widehat{C}_{n}(u, v),\notag\\
	T_{n,m}&=\sup_{(u,v)\in [0, 1]^2}\left|\widehat{C}_{n, m}(u, v)-\widehat{C}_{n, m}(v, u)\right|.
\end{align}
In what follows, the asymptotic behaviour of the three test statistics will be studied, including the asymptotic limit under both the null and alternative hypotheses.
\subsection{Asymptotic behaviour of the test statistics}

\n As it will be seen, the limits of the proposed test statistics are functional of the unknown underlying copula $ C $, therefore, some common assumptions are needed before going further.

\begin{assumption}\label{ass:1}
	Assume that first-order partial derivatives $\dot{C}_1(u, v)=\partial C(u, v)/\partial u, \dot{C}_{2}(u, v)=\partial C(u, v)/\partial v$ exist and are continuous, respectively, on the sets $(0, 1)\times [0, 1]$ and $[0, 1]\times (0, 1)$. 
\end{assumption}

\begin{assumption}\label{ass:2}
	Assume that second-order partial derivatives $\ddot{C}_{11}(u, v)=\partial^2 C(u, v)/\partial u^2, \ddot{C}_{21}(u, v)=\partial^2 C(u, v)/\partial u\partial v$ and $\ddot{C}_{22}(u, v)=\partial^2 C(u, v)/\partial v^2$ exist and are continuous, respectively, on the sets $(0, 1)\times [0, 1], (0, 1)\times (0, 1)$ and $[0, 1]\times (0, 1)$, and there exists a constant $K>0$ such that
	\begin{equation*}
		\left|\ddot{C}_{11}(u, v)\right|\le \frac{K}{u(1-u)},\quad 	\left|\ddot{C}_{21}(u, v)\right|\le K\min\left(\frac{1}{u(1-u)}, \frac{1}{v(1-v)}\right),\quad \left|\ddot{C}_{22}(u, v)\right|\le \frac{K}{v(1-v)}.
	\end{equation*} 
\end{assumption}

\n Let $\BB_{n, m}$ and $\bS_{n, m}$ denote, respectively, the empirical Bernstein copula process and symmetrised empirical Bernstein copula process,  namely, for all $(u, v)\in [0, 1]^2$
\begin{equation*}
	\BB_{n, m}(u, v)=\sqrt{n}\Big\{\widehat{C}_{n,m}(u, v)-C(u, v)\Big\}, \quad \bS_{n, m}(u, v)=\sqrt{n}\Big\{\widehat{C}_{n,m}(u, v)-\widehat{C}_{n, m}(v, u)\Big\}.
\end{equation*}
From the work of~\cite{Segers2017}, suppose $C$ satisfies Assumption~\ref{ass:1} and $m=cn^{\alpha}$ with $c>0, \alpha\geqslant 1$, then, in $\ell^{\infty}([0,1]^2)$ (the Banach space of all real-valued, bounded functions on $[0, 1]^2$, equipped with sup-norm),
\begin{equation}\label{eq:2022-08-25, 2:21PM}
	\BB_{n, m}(u, v)\rightsquigarrow \BB_C(u, v)=\mathbb{C}(u, v)-\dot{C}_1(u, v)\mathbb{C}(u,1)-\dot{C}_2(u, v)\mathbb{C}(1, v), \quad \text{as $n\rightarrow \infty$},
\end{equation}
where $\mathbb{C}$ is a $C$-Brownian bridge whose covariance function is 
\begin{equation*}
	\Gamma_{\mathbb{C}}(u, v, s, t)=C(\min(u, s), \min(v, t))-C(u, v)C(s, t),	
\end{equation*}
for $(u, v), (s, t)\in [0, 1]^2$ and $``\rightsquigarrow$'' stands for the weak convergence. Based on this result, the weak convergence of the symmetrised empirical Bernstein copula process is presented in the following theorem.

\begin{thm}\label{thm:1}
	Let $ C $ be a symmetric copula satisfying Assumption~\ref{ass:1}, in addition, if  $m=cn^{\alpha}$ with $c>0, \alpha\geqslant 1$, then as $n\rightarrow \infty$, $\bS_{n,m}$ converges weakly to a Gaussian process $\bS_C$ defined by
	\begin{equation*}
		\bS_C(u, v)=\mathbb{D}(u, v)-\dot{C}_1(u, v)\mathbb{D}(u,1)-\dot{C}_2(u,v)\mathbb{D}(1,v),
	\end{equation*}
	for all $(u, v)\in [0,1]^2$, where $\mathbb{D}$ is a $C$-Brownian bridge with covariance function given at each $(u, v) ,(s ,t)\in [0,1]^2$ by $\Gamma_{\mathbb{D}}(u, v, s, t)=2\{\Gamma_{\mathbb{C}}(u, v, s, t)-\Gamma_{\mathbb{C}}(u, v, t, s)\}$.
\end{thm}
\begin{proof}
	Using similar approach as in~\citet[Proof of Prosition 2]{Genest2012}, under the null hypothesis, one can write that,
	\begin{align*}
		\bS_{n,m}(u, v)& = \sqrt{n}\left\{\hat{C}_{n, m}(u, v)-\hat{C}_{n, m}(v, u)\right\}-       \sqrt{n}C(u, v)+\sqrt{n}C(v, u)\\
				&=\BB_{n, m}(u, v)-\BB_{n, m}(v, u).
	\end{align*}
	Further, it was shown in~\cite{Segers2017}, that $\BB_{n,m}\rightsquigarrow \BB_C$ as $n\rightarrow \infty$. Then, the desired result can be proven using continuous mapping theorem, namely, 
	\begin{equation*}
		\bS_{n,m}\rightsquigarrow \bS_C,
	\end{equation*}
	where $\bS_C(u,v)=\BB_C(u,v)-\BB_C(v,u)$.	 
\end{proof}

\bn Based on Theorem~\ref{thm:1} and the fact that the proposed test statistics in~\eqref{eq:Statest} are functional of $\bS_{n, m}$, their asymptotic behaviours under the null hypothesis can be established as follows.
\begin{thm}\label{thm:2}
	Let $ C $ be a symmetric copula satisfying Assumption~\ref{ass:1}, in addition, if  $m=cn^{\alpha}$ with $c>0, \alpha\geqslant 1$, then as $n \rightarrow \infty$,
	\begin{align*}
		nR_{n,m}=\int_{0}^1\int_{0}^1\big\{\bS_{n, m}(u, v)\big\}^2 \dif u \dif v&\rightsquigarrow \bL_R= \int_{0}^1\int_{0}^1\big\{\bS_C(u, v)\big\}^2 \dif u \dif v,\\
		nS_{n,m}=\int_{0}^1\int_0^1\big\{\bS_{n,m}(u,v)\big\}^2 \dif \widehat{C}_n(u,v)&\rightsquigarrow \bL_S=\int_{0}^1\int_{0}^1\big\{\bS_C(u,v)\big\}^2 \dif  C(u ,v),\\
		n^{1/2}T_{n,m}=\underset{(u,v)\in [0,1]^2}{\sup}|\bS_{n, m}(u, v)|&\rightsquigarrow \bL_T=\underset{(u,v)\in [0,1]^2}{\sup}\left|\bS_C(u, v)\right|.
	\end{align*}
\end{thm}
\begin{proof}[Proof of Theorem~\ref{thm:2}]
	The proof is postponed  to the appendix.
\end{proof}
\n More generally, the asymptotic behaviours of the test statistics under the alternative hypothesis are given as follows.

\begin{thm}\label{thm:3}
	Let $ C $ be a copula satisfying Assumption~\ref{ass:1}, in addition, if  $m=cn^{\alpha}$ with $c>0, \alpha\geqslant 1$, then as $n \rightarrow \infty$,
	\begin{align*}
		R_{n,m}&\xrightarrow{a.s.} R_C=\int_{0}^{1}\int_{0}^1\Big\{C(u, v)-C(v, u)\Big\}^2 \dif u \dif v,\\
		S_{n,m}&\xrightarrow{a.s.} S_C=\int_{0}^1\int_{0}^1\Big\{C(u, v)-C(v, u)\Big\}^2\dif C(u, v),\\
		T_{n,m}&\xrightarrow{a.s.} T_C=\sup_{(u,v)\in [0, 1]^2}\left|C(u, v)-C(v, u)\right|.
	\end{align*}
	Specifically, under the alternative hypothesis, as $n\to \infty$, 
	\begin{align*}
		nR_{n,m}\xrightarrow{a.s.} \infty,\quad nS_{n,m}\xrightarrow{a.s.} \infty,\quad \sqrt{n}T_{n,m}\xrightarrow{a.s.} \infty,
	\end{align*}
	which ensures the consistency of the proposed test statistics.
\end{thm}

\begin{proof}[Proof of Theorem~\ref{thm:3}]
	The proof is postponed  to the appendix.
\end{proof}

\begin{remark}
Note that,  the convergence in  Theorem~\ref{thm:1}-\ref{thm:3} also hold under Assumption~\ref{ass:1}-\ref{ass:2} with a weaker condition on $m$, that is, $m=cn^{\alpha}$ with $c>0, \alpha > 3/4$.
\end{remark}

\begin{lem}\label{prop:2023-08-24, 4:21PM}
	Suppose that $C$ satisfies Assumption~\ref{ass:1}-\ref{ass:2}, if $m=cn^{\alpha}$ with $c>0, \alpha > 3/4$, then, in $\ell^{\infty}([0, 1]^2)$, Equation~\eqref{eq:2022-08-25, 2:21PM} holds.
\end{lem}

\begin{proof}[Proof of Lemma~\ref{prop:2023-08-24, 4:21PM}]
	The proof is postponed  to the appendix.
\end{proof}

\section{Multiplier bootstrap\label{sec:Multiplier}}

\n As shown in the preceding section, the asymptotic limits of the statistics are functions of the unknown copula $C$, therefore it is impossible to compute valid P-values using standard Monte Carlo procedure directly. To overcome this issue, different bootstrap methods were developed, see for example~\cite{Genest2009} and \cite{Belalia2017b} among others. However, these approaches are computationally intensive, especially when the sample size is large. The multiplier bootstrap is an alternative methodology that mitigates the burden of computation by estimating the replicates of statistics under null hypothesis straightly. For the application of this method, the reader is directed to~\cite{Remillard2009},~\cite{Genest2012}, ~\cite{Bahraoui2018} and references therein. 

Following the multiplier procedure described in~\cite{Harder2017}, let $H\in \mathbb{N}$ and for each $h\in \{1, \ldots, H\}$, let $\bm{\xi}_n^{(h)}=\left(\xi_1^{(h)},\ldots,\xi_n^{(h)}\right)$ be a vector of independent random variables with unit mean and unit variance (taking $\xi_i^{(h)}\sim \rm{Exp}(1)$, $i=1,\ldots, n$).  Set  
\begin{equation*}
	C_{n,m}(u,v)=\frac{1}{n}\sum_{i=1}^{n}\sum_{k=0}^m\sum_{\ell=0}^m\II\left(U_i\le \frac{k}{m}, V_i\le \frac{\ell}{m}\right)P_{m,k}(u)P_{m,\ell}(v),
\end{equation*} 
where $\left(U_i, V_i\right)=\left(F_X\left(X_i\right), F_Y\left(Y_i\right)\right)$. Denote the sample mean of $\bm{\xi}_n^{(h)}$ by $\bar{\xi}_n^{(h)}$, then for each $h$, define
\begin{align*}
	\widetilde{\BB}_{n,m}(u, v)&=n^{1/2}\big\{C_{n,m}(u, v)-C(u, v)\big\}\\
	&=n^{1/2}\Bigg\{\frac{1}{n}\sum_{i=1}^n\sum_{k=0}^m\sum_{\ell=0}^m\bigg\{\II\left(U_i\le \frac{k}{m}, V_i\le \frac{\ell}{m}\right)-C(u, v)\bigg\} P_{m,k}(u)P_{m,\ell}(v)\Bigg\},
\end{align*}
and
\begin{align}\label{eq:2022-08-25, 2:08PM}
	\overline{\BB}_{n,m}^{(h)}(u, v) &=n^{1/2}\Bigg\{\frac{1}{n}\sum_{i=1}^n\sum_{k=0}^m\sum_{\ell=0}^m\left(\xi^{(h)}_i-\bar{\xi}_n^{(h)}\right)\II\left(\widehat{U}_i\le\frac{k}{m}, \widehat{V}_i\le \frac{\ell}{m}\right)P_{m,k}(u)P_{m,\ell}(v)\Bigg\}.
\end{align}
One can observe that $\widetilde{\BB}_{n,m}$ is the empirical Bernstein copula process when margins are known.  The following proposition states that $\widetilde{\BB}_{n,m}$ converges weakly to the $C$-Brownian bridge $\CC$ mentioned in~\eqref{eq:2022-08-25, 2:21PM}, moreover, $\overline{\BB}_{n,m}^{(h)}$ is a valid replicate of $\widetilde{\BB}_{n, m}$.

\begin{prop}\label{prop:1}
	Let $ C $ be a symmetric copula  satisfying Assumption~\ref{ass:1}-\ref{ass:2}, if 	$m=cn^{\alpha}$ with $c>0, \alpha > 3/4$,	then for $(u, v)\in [0, 1]^2$,
	\begin{enumerate}
		\item we have
		\begin{equation*}
			\widetilde{\BB}_{n,m}(u, v)
			\rightsquigarrow \CC(u ,v),
		\end{equation*}
		\item also, for each $h\in \{1, \ldots, H\}$,
		\begin{equation}\label{eq:(4)}
			\overline{\BB}_{n,m}^{(h)}(u, v) \rightsquigarrow \CC(u ,v).
		\end{equation}	
	\end{enumerate}
\end{prop}

\begin{proof}[Proof of Proposition~\ref{prop:1}]
	The proof is postponed  to the appendix.
\end{proof}

\n Hence the bootstrap replicates of $\BB_{n, m}(u, v)$ for $h\in \{1, \ldots, H\}$ are
\begin{equation}\label{eq:2023-05-10, 3:30PM}
	\BB_{n,m}^{(h)}(u, v)=\overline{\BB}^{(h)}_{n,m}(u, v)-\frac{\partial \widehat{C}_{n,m}(u, v)}{\partial u}\overline{\BB}^{(h)}_{n,m}(u, 1)-\frac{\partial \widehat{C}_{n,m}(u, v)}{\partial v}\overline{\BB}^{(h)}_{n,m}(1, v).
\end{equation}
The partial derivatives of the empirical Bernstein copula were studied in~\cite{Janssen2016}, more specifically, 
\begin{align*}
	\frac{\partial \widehat{C}_{n,m}(u, v)}{\partial u}&=m\sum_{k=0}^{m-1}\sum_{\ell=0}^m\bigg\{\widehat{C}_n\left(\frac{k+1}{m}, \frac{\ell}{m}\right)-\widehat{C}_n\left(\frac{k}{m}, \frac{\ell}{m}\right)\bigg\}P_{m-1, k}(u)P_{m, \ell}(v),\\
	\frac{\partial \widehat{C}_{n,m}(u, v)}{\partial v}&=m\sum_{k=0}^{m}\sum_{\ell=0}^{m-1}\bigg\{\widehat{C}_n\left(\frac{k}{m}, \frac{\ell+1}{m}\right)-\widehat{C}_n\left(\frac{k}{m}, \frac{\ell}{m}\right)\bigg\}P_{m, k}(u)P_{m-1, \ell}(v).
\end{align*}
Unlike the empirical copula, the partial derivatives of the empirical Bernstein copula can be calculated directly without any further approximations. The following proposition provides the uniform consistency of partial derivatives.

\begin{prop}\label{prop:2}
	Let $ C $ be a copula  satisfying Assumption~\ref{ass:1}-\ref{ass:2}, if $m=cn^{\alpha}$ with $c>0, \alpha > 3/4$ and $b_n=\kappa n^{-\beta}$ with $\kappa>0, \beta<{\alpha}/{2}$,
	then 
	\begin{align*}
		\sup_{\substack{v\in [0, 1],\\u\in [b_n, 1-b_n]}}\left|	\frac{\partial \widehat{C}_{n,m}(u, v)}{\partial u}-\dot{C}_1(u, v)\right|&=O\left(m^{1/2}n^{-1/2}(\log\log n)^{1/2}\right),\\
		\sup_{\substack{u\in [0, 1],\\v\in [b_n, 1-b_n]}}\left|	\frac{\partial \widehat{C}_{n,m}(u, v)}{\partial v}-\dot{C}_2(u, v)\right|&=O\left(m^{1/2}n^{-1/2}(\log\log n)^{1/2}\right),
	\end{align*}
	almost surely as $n\to \infty$.
\end{prop}

\begin{proof}[Proof of Proposition~\ref{prop:2}]
	The proof is postponed  to the appendix.
\end{proof}

The limiting behaviour of the multiplier replicates of the symmetrised empirical copula process are stated in the following theorem.

\begin{thm}\label{thm:4}
	Let $ C $ be a symmetric copula  satisfying Assumption~\ref{ass:1}-\ref{ass:2}, if $m=cn^{\alpha}$ with $3/4<\alpha<1$ and  $c>0$ , then for $(u, v)\in (0, 1)^2$,
	\begin{equation*}
		\left(\bS_{n,m}, \bS_{n,m}^{(1)}, \ldots, \bS_{n,m}^{(H)}\right)\rightsquigarrow \left(\bS_C, \bS_C^{(1)}, \ldots, \bS_C^{(H)}\right), 
	\end{equation*}
	as $n\rightarrow \infty$, where for each $h\in \{1, \ldots, H\}$,
	\begin{align*}
		\bS_{n,m}^{(h)}(u, v)&=\BB_{n,m}^{(h)}(u, v)-\BB_{n,m}^{(h)}(v, u)\\
		&=\overline{\bS}^{(h)}_{n,m}(u, v)+\frac{\partial \widehat{C}_{n,m}(u, v)}{\partial u}\overline{\bS}^{(h)}_{n,m}(u, 1)-\frac{\partial \widehat{C}_{n,m}(u, v)}{\partial v}\overline{\bS}^{(h)}_{n,m}(1, v),
	\end{align*}
	and $	\overline{\bS}^{(h)}_{n,m}(u, v)=\overline{\BB}^{(h)}_{n,m}(u, v)-\overline{\BB}^{(h)}_{n,m}(v, u)$.
	
\end{thm}
\begin{proof}[Proof of Theorem~\ref{thm:4}]
Given that under the null hypothesis
\begin{equation*}
	\dot{C}_1(u, v)=\dot{C}_2(v, u),
\end{equation*}
the proof is a direct application of Proposition~\ref{prop:2} and the continuous mapping theorem.
\end{proof}
\bn Combining Theorems~\ref{thm:2} and~\ref{thm:4}, the asymptotic properties of replicates of statistics defined by~\eqref{eq:Statest} are established in the following corollary.
\begin{corollary}\label{Cor:1}
	Let $ C $ be a symmetric copula  satisfying Assumption~\ref{ass:1}-\ref{ass:2}, if $m=cn^{\alpha}$ with $3/4< \alpha<1$ and $c>0$,  then as $n\rightarrow \infty$, for all $H\in \NN$, one has 
	\begin{align*}
		\left(nR_{n,m}, nR_{n,m}^{(1)}, \ldots , nR_{n,m}^{(H)}\right)&\rightsquigarrow \left(\bL_R, \bL_R^{(1)}, \ldots, \bL_R^{(H)}\right),\\
		\left(nS_{n,m}, nS_{n,m}^{(1)}, \ldots , nS_{n,m}^{(H)}\right)&\rightsquigarrow \left(\bL_S, \bL_S^{(1)}, \ldots, \bL_S^{(H)}\right),\\
		\left(n^{1/2}T_{n,m}, n^{1/2}T_{n,m}^{(1)}, \ldots , n^{1/2}T_{n,m}^{(H)}\right)&\rightsquigarrow \left(\bL_T, \bL_T^{(1)}, \ldots, \bL_T^{(H)}\right).\\
	\end{align*}
\end{corollary}

\begin{remark}
Note that, under Assumption 1-2, the proposed statistics together with multiplier bootstraps can be valid under the same magnitude requirement of $m$, which is $m=cn^{\alpha}$ with $c>0, 3/4 <\alpha< 1$.
\end{remark}

\n It follows directly from Corollary~\ref{Cor:1} that P-values of the proposed tests can be computed as
\begin{equation*}
	\frac{1}{H}\sum_{h=1}^{H}\II\left(R_{n,m}^{(h)}\geqslant R_{n,m}\right),\quad 	\frac{1}{H}\sum_{h=1}^{H}\II\left(S_{n,m}^{(h)}\geqslant S_{n,m}\right), \quad 	\frac{1}{H}\sum_{h=1}^{H}\II\left(T_{n,m}^{(h)}\geqslant T_{n,m}\right).
\end{equation*}
This approach will be employed to obtain the empirical level and power as shown in the next section.

\section{Finite sample performance \label{sec:Sim}}
\n The finite sample performance of the proposed testing procedure is investigated in this section through a Monte Carlo experiment. All the tests were conducted with $500$ repetitions under $5\%$ nominal level using $ H=200 $ multiplier replicates, also since the statistics shown in~\thmref{thm:2} involve integration, a discrete approximation is applied with size of integration grid $N=20$ (i.e.  $20\times 20$ points on $[0, 1]^2$).
\subsection{Comparison with the empirical copula-based tests}

 To assess the improvement of the empirical Bernstein copula process-based tests defined in~\eqref{eq:(4)} and denoted by $\{R_{n, m}, S_{n, m}, T_{n, m}\}$,  the empirical level and power of the proposed tests were compared with the tests based on the empirical copula process in~\cite{Genest2012} denoted by $\{R_n, S_n, T_n\}$. For the empirical size, samples with sizes $n=\{50, 100, 200\}$ were generated from the Gaussian, Clayton, Gumbel-Hougaard and Frank copulas with Kendall's tau $\tau = 0.25$.  Since the Bernstein order $m$ is dependent on the sample size, $m$ is chosen using  $m= \lfloor c n^{\alpha} \rfloor =\{9, 13, 15\}$, where $\alpha =  4/5$  and $c = \{0.40, 0.33, 0.22\}$.

 \n The empirical level of the tests are presented in Table~\ref{tab:1}. Evidently, a significant portion of the Bernstein tests falls below their designated nominal level; however, their performance surpasses that of the empirical tests.

\begin{table}[H]
	\centering
	\caption{Level ($ \% $) of the tests of $ \mathscr{H}_0 $ based on $ \{ R_n, S_n, T_n\} $ and $ \{ R_{n,m}, S_{n,m}, T_{n,m}\} $, as estimated from $ 500 $ samples from symmetric copulas using $ H=200 $ multiplier bootstrap replicates with Bernstein order $m=\{9, 13, 15\}$ associated with sample size $n=\{50, 100, 200\}$.}
	\small
	\begin{tabular}[t]{lccccccc}
		\toprule[1.5pt]
		Model&$(n,m)$&$R_n$&$R_{n,m}$ &$S_n$&$S_{n,m}$&$T_n$&$T_{n,m}$ \\
		\midrule
		\multirow{3}{*}{$\rm{Gaussian}$}&$(50, 9)$&$1.2$&$2.0$ &$3.2$&$3.4$&$1.0$&$2.0$ \\
		&$(100, 13)$&$2.6$&$4.6$ &$4.4$&$5.2$&$2.8$&$4.0$ \\
		&$(200,15)$&$2.8$&$3.4$ &$1.8$&$3.6$&$7.4$&$3.8$ \\
		\midrule
		\multirow{3}{*}{$\rm{Clayton}$}&$(50, 9)$&$1.2$&$3.2$ &$4.2$&$4.2$&$1.0$&$2.2$ \\
		&$(100, 13)$&$1.0$&$3.6$ &$3.6$&$3.8$&$1.8$&$2.6$ \\
		&$(200,15)$&$4.4$&$4.6$ &$4.2$&$5.2$&$6.2$&$4.2$ \\	
		\midrule
		\multirow{2.5}{*}{$\rm{Gumbel}$}&$(50,9)$&$0.8$&$ 3.2$ &$3.4$&$3.8$&$0.4$&$1.6$ \\
			\multirow{1.5}{*}{$\rm{Hougaard}$}
		&$(100,13)$&$1.8$&$4.0$ &$2.8$&$3.4$&$1.8$&$2.8$ \\
		&$(200,15)$&$3.0$&$4.2$ &$2.6$&$4.0$&$6.2$&$4.0$\\ 
		\midrule
		\multirow{3}{*}{$\rm{Frank}$}&$(50, 9)$&$1.4$&$ 2.4$ &$3.4$&$3.4$&$1.6$&$2.2$ \\
		&$(100, 13)$&$2.6$&$2.6$ &$2.2$&$2.6$&$1.6$&$2.8$ \\
		&$(200,15)$&$3.6$&$3.8$ &$4.0$&$4.4$&$6.4$&$4.0$ \\
		\bottomrule[1.5pt]
	\end{tabular}
	
	\label{tab:1}
\end{table}

\n To study the power of the considered tests, samples of sizes $n=100$ and $n=200$ were generated from the Gaussian (GA), Frank (FR), Gumbel-Hougaard (GU)
and Student (ST) copulas, made asymmetric using the Khoudraji’s device~\citep{Khoudraji:1995}. The Khoudraji's device is defined as 
\begin{equation*}
	K_{\delta}(u, v)=u^{\delta}C(u^{1-\delta}, v),
\end{equation*}
and is implemented in the \textbf{R}~\citep{R} package \verb+copula+ by~\cite{Copula}. 

\n Different values of the shape parameter $\delta=\{1/4, 1/2, 3/4\}$ as well as various values of Kendall's tau $\tau=\{0.5, 0.7, 0.9\}$  were considered to assess their influence on the power. A quick inspection of  Tables~\ref{table:2} and~\ref{table:3}, one can observe that for large value of $\tau$ or $n$, the power of the tests increase, and they reach their maximum at $\delta=1/2$. Under all  circumstances, the proposed tests outperform the empirical copulas tests or their differences are negligible. Moreover, significant improvements are discerned in the effectiveness of $T_{n}$, especially when $\tau$ assumes large values. This implies that, in such cases, Bernstein smoothing demonstrates greater efficacy for the Kolmogorov-Smirnov type statistic $T_n$ compared to the Cram\'{e}r-von Mises type statistics $R_n$ and $S_n$.  Finally, by combining the tables of level and power, it is safe to claim that the power of the tests is not affected by the empirical levels.

\begin{table}[H]
	\centering
	\caption{Power ($ \% $) of the tests of $ \mathscr{H}_0 $ based on $ \{ R_n, S_n, T_n\} $ and $ \{ R_{n,m}, S_{n,m}, T_{n,m}\} $, as estimated from $ 500 $ samples from asymmetric copulas using $ H=200 $ multiplier bootstrap replicates with Bernstein order $m=13$ associated and sample size $n= 100$.}
	\resizebox{13cm}{!}{
		\begin{tabular}{lcccccccc}
			\toprule[1.5pt]
			Model&$\delta$&$\tau$&$R_n$&$R_{n,m}$ &$S_n$&$S_{n,m}$&$T_n$&$T_{n,m}$ \\
			\midrule
			\multirow{9}{*}{$K_{\delta}^{GA}$}&\multirow{3}{*}{$1/4$}&${0.5}$&$5.2$&$11.2$ &$7.0$&$11.0$&$5.0$&$7.4$ \\
			&&$0.7$&$31.4$&$48.2$ &$46.6$&$51.4$&$17.0$&$31.6$ \\
			&&$0.9$&$90.2$&$92.8$ &$99.8$&$98.2$&$57.0$&$71.6$ \\
			&\multirow{3}{*}{$1/2$}&$0.5$&$15.6$&$27.8$ &$18.8$&$29.6$&$9.0$&$18.2$ \\
			&&$0.7$&$70.2$&$83.6$ &$75.8$&$84.0$&$38.4$&$67.8$ \\
			&&$0.9$&$99.6$&$100.0$ &$100.0$&$100.0$&$81.0$&$98.8$ \\
			&\multirow{3}{*}{$3/4$}&$0.5$&$9.4$&$18.2$ &$13.8$&$18.2$&$6.2$&$12.2$ \\
			&&$0.7$&$42.8$&$62.2$ &$51.0$&$62.0$&$23.2$&$47.4$ \\
			&&$0.9$&$74.6$&$88.4$ &$85.2$&$88.2$&$44.4$&$74.0$ \\
			\midrule
			\multirow{9}{*}{$K_{\delta}^{FR}$}&\multirow{3}{*}{$1/4$}&$0.5$&$6.0$&$14.6$ &$11.8$&$14.0$&$5.0$&$8.6$ \\
			&&${0.7}$&$38.8$&$55.4$ &$54.4$&$56.6$&$22.0$&$36.2$ \\
			&&$0.9$&$90.0$&$90.0$ &$100.0$&$97.2$&$60.2$&$72.2$ \\
			&\multirow{3}{*}{$1/2$}&$0.5$&$21.0$&$31.6$ &$26.6$&$32.8$&$16.4$&$25.8$ \\
			&&$0.7$&$75.8$&$87.4$ &$84.8$&$89.0$&$46.4$&$76.4$ \\
			&&$0.9$&$99.8$&$99.8$ &$100.0$&$100.0$&$82.0$&$98.4$ \\
			&\multirow{3}{*}{$3/4$}&$0.5$&$12.2$&$20.0$ &$14.8$&$20.6$&$8.8$&$14.6$ \\
			&&$0.7$&$38.4$&$56.8$ &$46.6$&$54.8$&$20.2$&$44.0$ \\
			&&$0.9$&$71.0$&$88.4$ &$81.8$&$88.6$&$36.6$&$70.2$ \\
			\midrule
			\multirow{9}{*}{$K_{\delta}^{GU}$}&\multirow{3}{*}{$1/4$}&${0.5}$&$7.4$&$15.2$ &$9.8$&$14.8$&$8.0$&$10.6$ \\
			&&$0.7$&$35.2$&$50.6$ &$48.2$&$54.2$&$18.4$&$31.4$ \\
			&&$0.9$&$90.6$&$86.8$ &$99.8$&$97.0$&$58.2$&$70.6$ \\
			&\multirow{3}{*}{$1/2$}&$0.5$&$20.8$&$32.2$ &$26.6$&$35.2$&$10.2$&$24.0$ \\
			&&$0.7$&$80.6$&$87.6$ &$87.4$&$89.6$&$47.6$&$77.2$ \\
			&&$0.9$&$100.0$&$100.0$ &$100.0$&$100.0$&$84.6$&$98.2$ \\
			&\multirow{3}{*}{$3/4$}&$0.5$&$21.8$&$32.6$ &$24.0$&$32.2$&$11.4$&$21.8$ \\
			&&$0.7$&$56.0$&$73.4$ &$62.0$&$73.0$&$27.6$&$57.0$ \\
			&&$0.9$&$77.0$&$86.8$ &$86.8$&$88.0$&$44.0$&$75.4$ \\		
			\midrule
			\multirow{9}{*}{$K_{\delta}^{ST}$}&\multirow{3}{*}{$1/4$}&$0.5$&$5.8$&$14.4$ &$11.0$&$17.2$&$7.4$&$9.2$ \\
			&&$0.7$&$36.8$&$48.0$ &$51.4$&$54.0$&$18.6$&$32.6$ \\
			&&$0.9$&$88.4$&$88.8$ &$99.4$&$95.6$&$55.0$&$72.6$ \\
			&\multirow{3}{*}{$1/2$}&$0.5$&$16.0$&$29.8$ &$24.4$&$30.2$&$11.0$&$20.8$ \\
			&&$0.7$&$71.0$&$82.6$ &$80.4$&$84.4$&$38.6$&$68.4$ \\
			&&$0.9$&$99.8$&$100.0$ &$100.0$&$100.0$&$78.2$&$97.8$ \\
			&\multirow{3}{*}{$3/4$}&$0.5$&$12.4$&$23.6$ &$15.4$&$22.4$&$8.6$&$15.8$ \\
			&&$0.7$&$40.8$&$57.8$ &$54.0$&$60.8$&$22.4$&$43.8$ \\
			&&$0.9$&$75.4$&$87.8$ &$84.0$&$88.6$&$42.0$&$73.2$ \\
			\bottomrule[1.5pt]
	\end{tabular}}
	\label{table:2}
\end{table}

\begin{table}[H]
	\centering
	\caption {Power ($ \% $) of the tests of $ \mathscr{H}_0 $ based on $ \{ R_n, S_n, T_n\} $ and $ \{ R_{n,m}, S_{n,m}, T_{n,m}\} $, as estimated from $ 500 $ samples from asymmetric copulas using $ H=200 $ multiplier bootstrap replicates with Bernstein order $m=15$ associated and sample size $n= 200$.}
	\resizebox{13cm}{!}{
		\begin{tabular}[t]{lcccccccc}
			\toprule[1.5pt]
			Model&$\delta$&$\tau$&$R_n$&$R_{n,m}$ &$S_n$&$S_{n,m}$&$T_n$&$T_{n,m}$ \\
			\midrule
			\multirow{9}{*}{$K_{\delta}^{GA}$}&\multirow{3}{*}{$1/4$}&${0.5}$&$15.0$&$24.0$ &$14.8$&$22.8$&$13.4$&$14.6$ \\
			&&$0.7$&$81.8$&$88.4$ &$82.6$&$89.4$&$59.8$&$70.4$ \\
			&&$0.9$&$100.0$&$100.0$ &$100.0$&$100.0$&$100.0$&$100.0$ \\
			&\multirow{3}{*}{$1/2$}&$0.5$&$37.0$&$45.4$ &$36.0$&$47.8$&$27.0$&$33.0$ \\
			&&$0.7$&$98.8$&$98.8$ &$99.0$&$99.8$&$88.2$&$96.6$ \\
			&&$0.9$&$100.0$&$100.0$ &$100.0$&$100.0$&$100.0$&$100.0$ \\
			&\multirow{3}{*}{$3/4$}&$0.5$&$32.4$&$44.2$ &$32.6$&$43.0$&$27.4$&$32.8$ \\
			&&$0.7$&$87.8$&$93.2$ &$90.0$&$94.0$&$67.6$&$84.8$ \\
			&&$0.9$&$99.8$&$99.8$ &$99.8$&$99.8$&$92.2$&$98.6$ \\
			\midrule
			\multirow{9}{*}{$K_{\delta}^{FR}$}&\multirow{3}{*}{$1/4$}&$0.5$&$25.2$&$34.4$ &$25.4$&$35.4$&$24.2$&$25.4$ \\
			&&${0.7}$&$93.2$&$95.0$ &$95.0$&$95.6$&$80.2$&$84.2$ \\
			&&$0.9$&$100.0$&$100.0$ &$100.0$&$100.0$&$100.0$&$99.4$ \\
			&\multirow{3}{*}{$1/2$}&$0.5$&$52.2$&$60.8$ &$52.0$&$61.4$&$41.2$&$48.2$ \\
			&&$0.7$&$99.6$&$100.0$ &$99.6$&$100.0$&$95.4$&$99.0$ \\
			&&$0.9$&$100.0$&$100.0$ &$100.0$&$100.0$&$100.0$&$100.0$ \\
			&\multirow{3}{*}{$3/4$}&$0.5$&$32.8$&$41.8$ &$34.4$&$42.6$&$26.4$&$34.2$ \\
			&&$0.7$&$79.2$&$88.4$ &$80.8$&$89.6$&$62.8$&$78.2$ \\
			&&$0.9$&$99.2$&$99.6$ &$99.6$&$99.6$&$91.6$&$99.0$ \\
			\midrule
			\multirow{9}{*}{$K_{\delta}^{GU}$}&\multirow{3}{*}{$1/4$}&${0.5}$&$21.2$&$30.6$ &$21.8$&$29.2$&$20.2$&$20.6$ \\
			&&$0.7$&$89.4$&$92.6$ &$89.8$&$92.2$&$68.4$&$79.2$ \\
			&&$0.9$&$100.0$&$100.0$ &$100.0$&$100.0$&$99.6$&$99.6$ \\
			&\multirow{3}{*}{$1/2$}&$0.5$&$59.6$&$67.4$ &$57.8$&$67.6$&$45.8$&$55.6$ \\
			&&$0.7$&$99.6$&$99.8$ &$99.6$&$99.8$&$94.4$&$98.4$ \\
			&&$0.9$&$100.0$&$100.0$ &$100.0$&$100.0$&$100.0$&$100.0$ \\
			&\multirow{3}{*}{$3/4$}&$0.5$&$54.6$&$66.2$ &$56.0$&$65.2$&$41.4$&$54.4$ \\
			&&$0.7$&$94.2$&$96.4$ &$94.8$&$96.6$&$77.2$&$91.2$ \\
			&&$0.9$&$99.2$&$99.4$ &$99.2$&$99.2$&$92.6$&$98.8$ \\		
			\midrule
			\multirow{9}{*}{$K_{\delta}^{ST}$}&\multirow{3}{*}{$1/4$}&$0.5$&$20.8$&$25.6$ &$21.4$&$27.6$&$16.4$&$20.2$ \\
			&&$0.7$&$93.2$&$95.0$ &$95.0$&$95.6$&$80.2$&$84.2$ \\
			&&$0.9$&$100.0$&$100.0$ &$100.0$&$100.0$&$100.0$&$99.4$ \\
			&\multirow{3}{*}{$1/2$}&$0.5$&$52.2$&$60.8$ &$52.0$&$61.4$&$41.2$&$48.2$ \\
			&&$0.7$&$99.6$&$100.0$ &$99.6$&$100.0$&$95.4$&$99.0$ \\
			&&$0.9$&$100.0$&$100.0$ &$100.0$&$100.0$&$100.0$&$100.0$ \\
			&\multirow{3}{*}{$3/4$}&$0.5$&$32.8$&$41.8$ &$34.4$&$42.6$&$26.4$&$34.2$ \\
			&&$0.7$&$79.2$&$88.4$ &$80.8$&$89.6$&$62.8$&$78.2$ \\
			&&$0.9$&$99.2$&$99.6$ &$99.6$&$99.6$&$91.6$&$99.0$ \\
			\bottomrule[1.5pt]
	\end{tabular}}
	
	\label{table:3}
\end{table}
\newpage
Selecting the Bernstein order for each test is beyond the scope of this paper. However, to examine the effect of the Bernstein order $m$ on the power of the proposed procedures, one can plot the graph of the power as a function of $m$. Figure~\ref{fig:Order} depicts the power of the proposed tests as function of  the Bernstein polynomial order $ m $ for asymmetric Gumbel-Hougaard copula model with shape parameter $\delta=3/4$ and Kendall's tau  $\tau=0.7$. From that figure, it can be seen that $R_{n, m}$ and $S_{n, m}$ experience a similar pattern as they are both Cram\'er-von Mises type statistics and the power seems to be stable for increasing $m$. Unexpectedly, the power of the test based on $T_{n, m}$ decreases as $m$ goes up. Overall, they outperform the empirical copula tests when $m$ is not extremely small. It is noted that other asymmetric copulas have shown almost the same pattern, and are not reported here. A practical way to select an appropriate Bernstein order $m$ will be discussed in the next subsection.  

\begin{figure}[H]
	\centering
	\begin{subfigure}{.5\textwidth}
		\includegraphics[width=3.1in, height=3.1in]{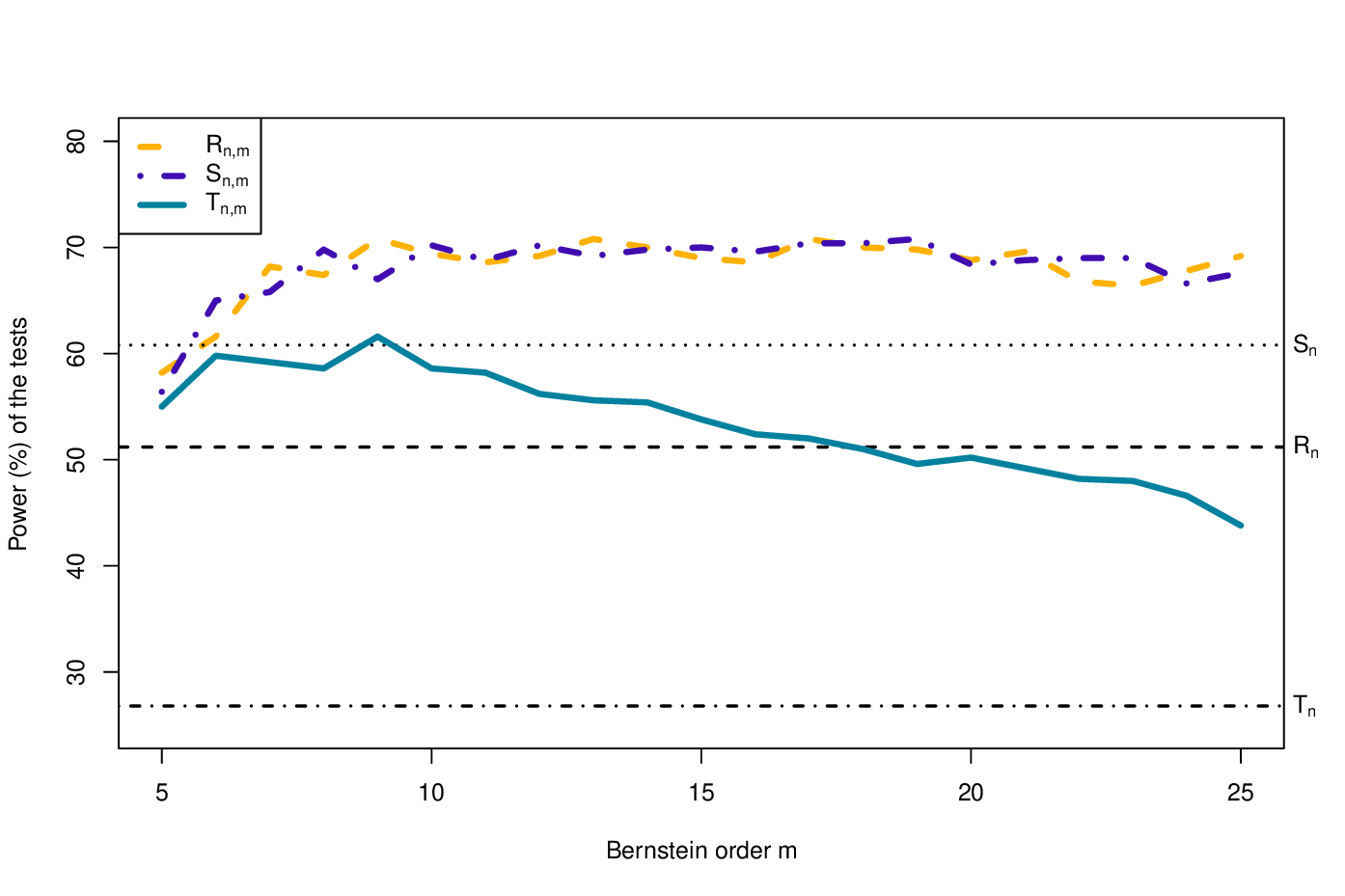}
		\label{fig:n=100}
	\end{subfigure}%
	\begin{subfigure}{.5\textwidth}
		\includegraphics[width=3.1in, height=3.1in]{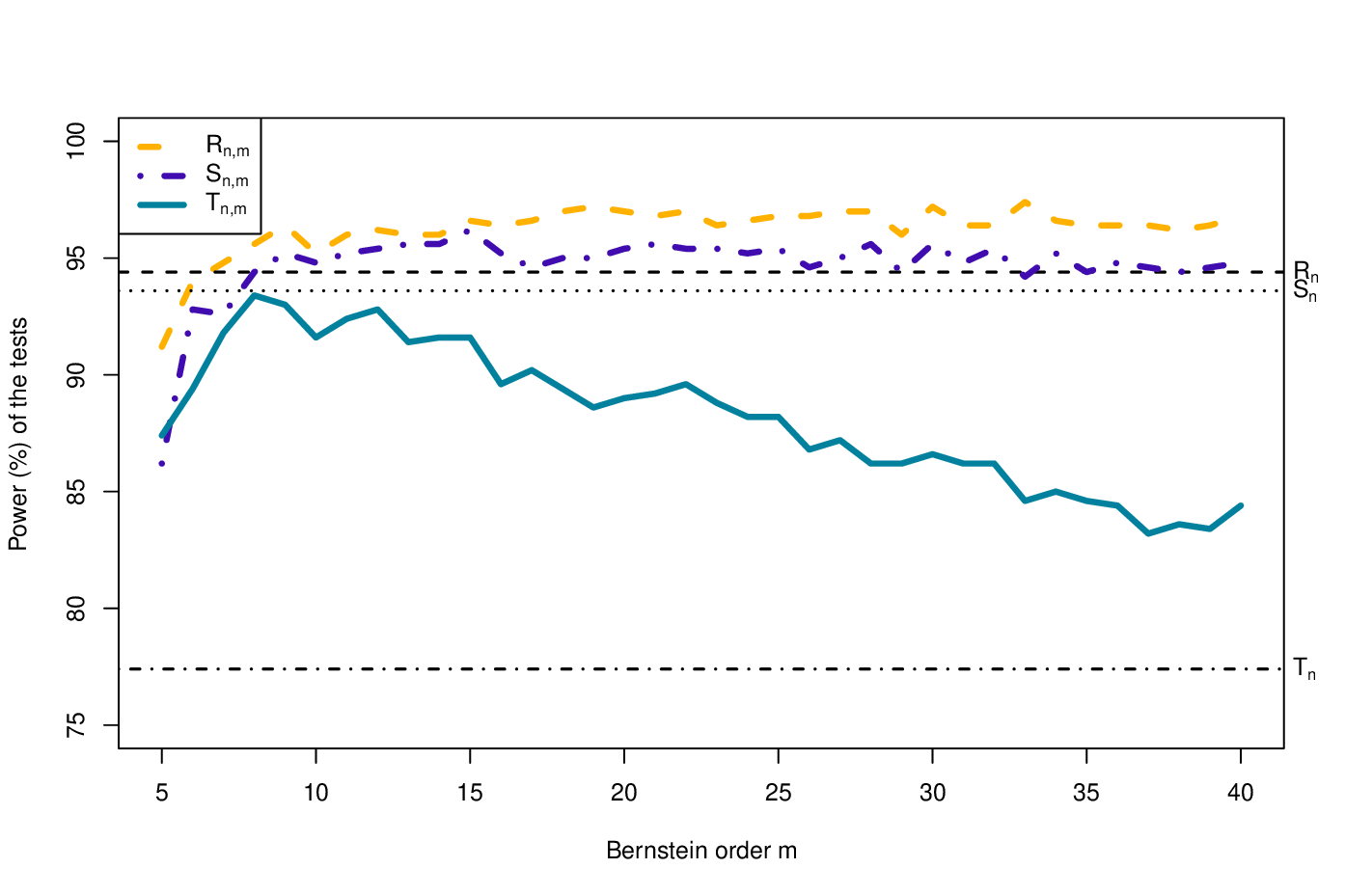}
		\label{fig:n=200}
	\end{subfigure}
	\caption{The power ($\%$) of the proposed tests versus the Bernstein polynomial order $ m $. \textbf{Left:} Gumbel-Hougaard copula with sample size $n=100$; \textbf{Right:}  Gumbel-Hougaard copula with sample size $n=200$.} 
	\label{fig:Order}
\end{figure}

\subsection{Comparison with empirical beta copula-based test}

The method of studying the effect of the Bernstein order $m$ in the preceding subsection is computationally expensive. In practice, one can apply the method suggested in~\cite{Segers2017} to reduce the size of the grid of the Bernstein order $m.$ This recommendation will be used in order to compare the empirical power of the proposed testing procedure to the test based on the empirical beta copula denoted by $R_n^{\beta}$ in~\cite{Kiriliouk2021}. It was shown in~\citet[Table 2.8]{Kiriliouk2021} that the smoothed beta bootstrap test based on $R_n^{\beta}$ has slightly higher power in almost all cases. In Table~\ref{table:4} the same configuration for Clayton copula made asymmetric by Khoudraji's device as in~\citet[Table 2.8]{Kiriliouk2021} was used. The table highlights the advantage of the empirical Bernstein copula-based tests using multiplier bootstraps on the smoothed beta bootstrap in almost all cases. 

\begin{remark}
	Note that, the selection method of the Bernstein order $m$ recommended in~\cite{Jansen2012} is not valid for the proposed multiplier bootstraps. Because this method of selection depends on the first and second partial derivatives of the underlying copula function $C$. Whereas for most copula models, these partial derivatives do not exist at the end points $(u, 1), (1, u), (1, v), (v, 1)$. Therefore, it fails to apply this approach for the second and third terms in Equation~\eqref{eq:2023-05-10, 3:30PM}. It is also noted that this method of selection $ m $ is valid only for interior points of $ [0,1]^2 $.
\end{remark}

\begin{table}[H]
	\centering
	\caption{Power ($ \% $) of the tests of $ \mathscr{H}_0 $ based on $ \{ R_n, S_n, T_n\} $, $ \{ R_{n,m}, S_{n,m}, T_{n,m}\} $ and $ R_n^{\beta} $, as estimated from $ 500 $ samples from Clayton made asymmetric using Khoudraji's device with $ H =200 $ multiplier replicates.}
	\resizebox{17cm}{!}{
		\begin{tabular}[t]{lccccccccc}
			\toprule[1.5pt]
			(n, m)&$\delta$&$\tau$&$R_n$&$R_{n,m}$ &$S_n$&$S_{n,m}$ &$T_n$&$T_{n,m}$ &$R_n^{\beta}$\\
			\midrule
			\multirow{9}{*}{$\left(50,  50/5\right)$}&\multirow{3}{*}{$1/4$}&${0.25}$&$1.0$&$3.6$ &$2.2$&$\textbf{3.8}$&$0.6$&$2.4$ &$3.6$\\
			&&$0.50$&$2.0$&$5.4$ &$5.4$&$\textbf{5.8}$&$1.2$&$1.8$&$5.4$ \\
			&&$0.75$&$14.8$&$21.0$ &$\textbf{41.8}$&$28.8$&$2.4$&$12.0$&$22.2$ \\
			&\multirow{3}{*}{$1/2$}&$0.25$&$0.8$&$\textbf{5.0}$ &$4.2$&$4.8$&$0.4$&$2.4$&$4.8$ \\
			&&$0.50$&$3.0$&$7.2$ &$6.4$&$\textbf{7.4}$&$0.6$&$3.8$&$7.0$ \\
			&&$0.75$&$24.8$&$45.8$ &$48.2$&$\textbf{52.0}$&$5.2$&$29.6$&$44.0$ \\
			&\multirow{3}{*}{$3/4$}&$0.25$&$0.8$&$\textbf{4.4}$ &$3.6$&$3.8$&$0.4$&$1.8$&$2.8$ \\
			&&$0.50$&$2.0$&$\textbf{6.0}$ &$5.6$&$5.8$&$1.0$&$2.6$&$5.8$ \\
			&&$0.75$&$9.8$&$\textbf{21.4}$& $15.0$&$\textbf{21.4}$&$2.4$&$11.2$&$18.0$ \\
			
			\midrule
			\multirow{9}{*}{$\left(100,  100/5\right)$}&\multirow{3}{*}{$1/4$}&${0.25}$&$1.8$&$\textbf{4.8}$ &$2.4$&$4.4$&$2.6$&$3.8$ &$4.0$\\
			&&$0.50$&$6.4$&$10.2$ &$9.2$&$\textbf{11.0}$&$5.0$&$5.2$ &$9.8$\\
			&&$0.75$&$51.4$&$62.6$ &$\textbf{77.6}$&$71.2$&$25.2$&$40.6$&$64.4$ \\
			&\multirow{3}{*}{$1/2$}&$0.25$&$1.6$&$3.6$&$2.2$ &$\textbf{4.0}$&$2.2$&$2.0$&$3.4$ \\
			&&$0.50$&$8.8$&$16.8$ &$14.0$&$\textbf{17.6}$&$5.8$&$10.4$&$14.0$ \\
			&&$0.75$&$70.6$&$83.8$&$82.2$&$\textbf{86.2}$&$41.6$&$66.0$&$82.8$ \\
			&\multirow{3}{*}{$3/4$}&$0.25$&$3.2$&$5.6$ &$5.2$&$\textbf{6.2}$&$2.6$&$3.4$&$4.4$ \\
			&&$0.50$&$6.2$&$10.6$ &$10.4$&$\textbf{12.2}$&$4.8$&$6.0$ &$10.0$\\
			&&$0.75$&$30.0$&$40.2$&$36.4$&$\textbf{42.2}$&$20.2$&$28.6$&$38.2$\\		
			\bottomrule[1.5pt]
	\end{tabular}}
	
	\label{table:4}
\end{table}

\section{Real data application \label{sec: realdata}}

\subsection{Ocean data application}

\n A simple illustration was carried out on the ocean data (at south Kodiak, Station $ 46066 $) from the \href{https://www.ndbc.noaa.gov/}{National Data Buoy Center (NDBC), US}. There are three variables of interest: WVHT (significant wave height in meter), APD (average wave period in second) and WSPD (wind speed in meter per second) during the winter season of $ 2015 $ (from November $ 2014 $ to February $ 2015 $)  with $2855$ observations. Although the data was modelled using asymmetric copulas in~\cite{Zhang2018}, they did not provide a formal statistical test to justify the asymmetric structure therein. Here, the justification is shown using the empirical tests and the proposed Bernstein tests as in Table~\ref{table:5}.

Rank plots are shown in Figure~\ref{fig:scatterplot}, one can notice that APD and WSPD are positively correlated with WVHT according to the Spearman's rho. Further, as presented in Figure~\ref{fig:rankplot}, these pairs most likely have an asymmetric dependence structure.

\begin{figure}[H]
	\centering
	\includegraphics[width=1\linewidth]{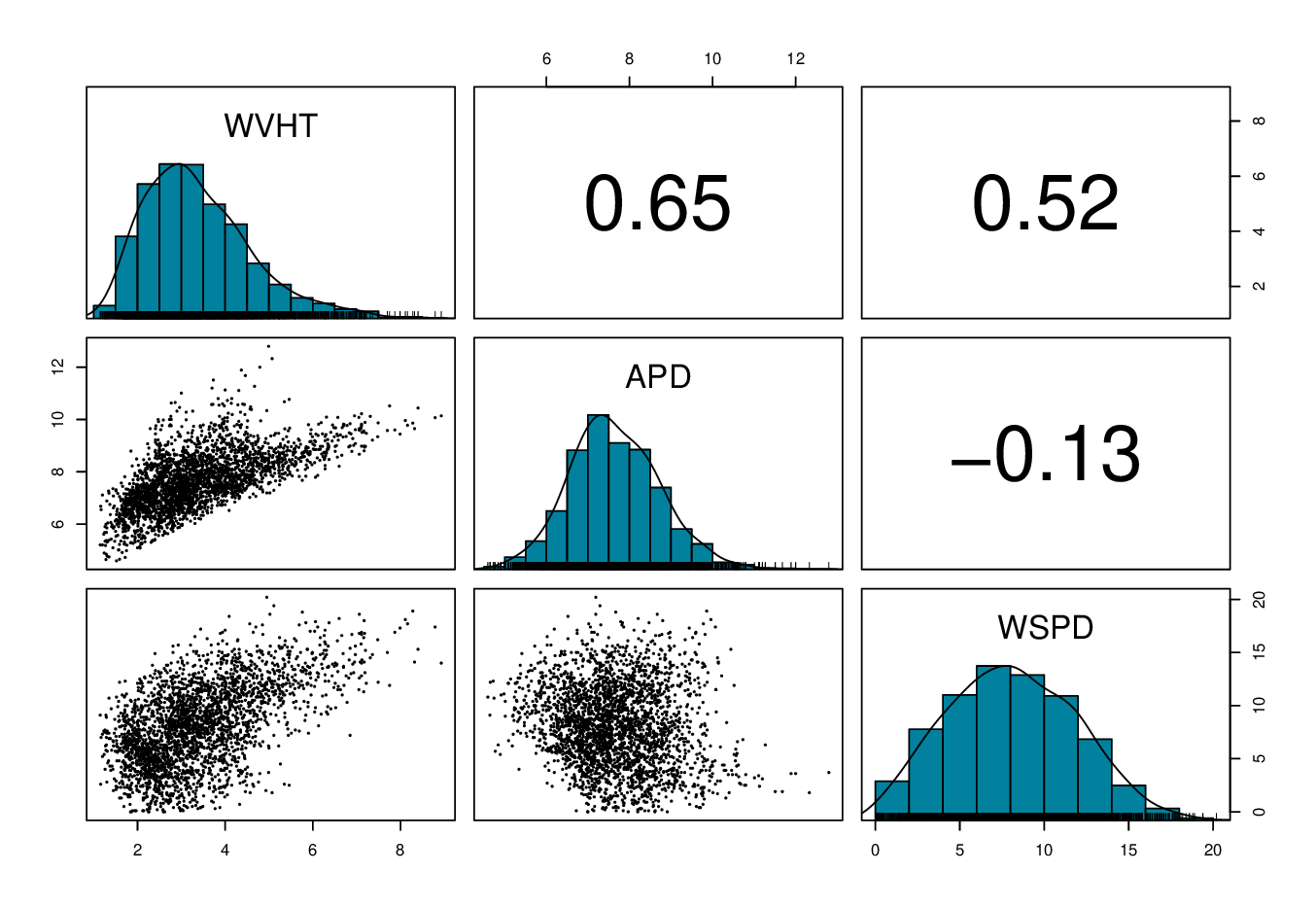}
	\caption{Scatter plot and Spearman's rho between WVHT, APD and WSPD.}
	\label{fig:scatterplot}
\end{figure}

\begin{figure}[H]
	\centering
	\begin{subfigure}{.5\textwidth}
		\includegraphics[width=3in, height=3in]{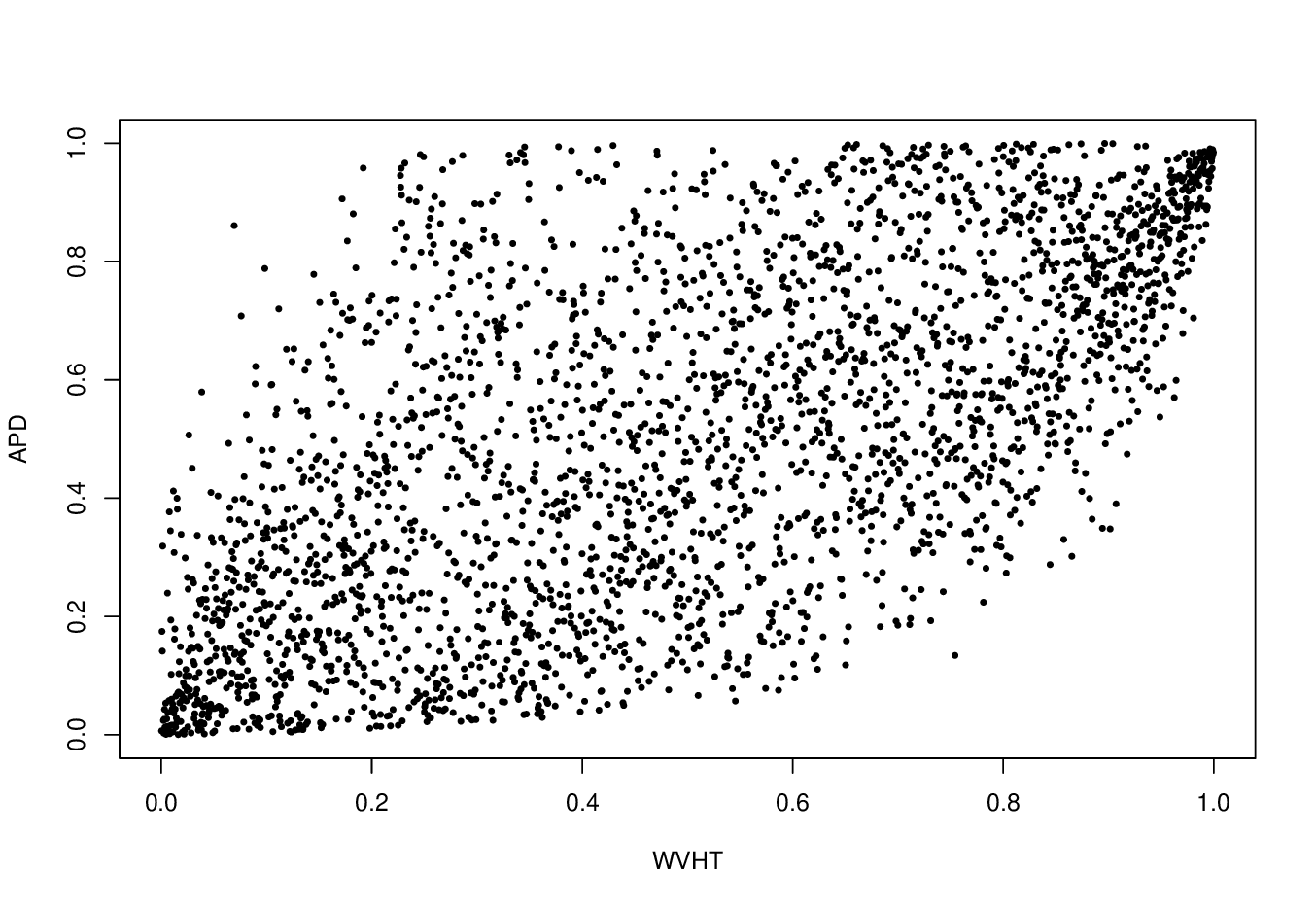}
		\label{fig:period}
	\end{subfigure}%
	\begin{subfigure}{.5\textwidth}
		\includegraphics[width=3in, height=3in]{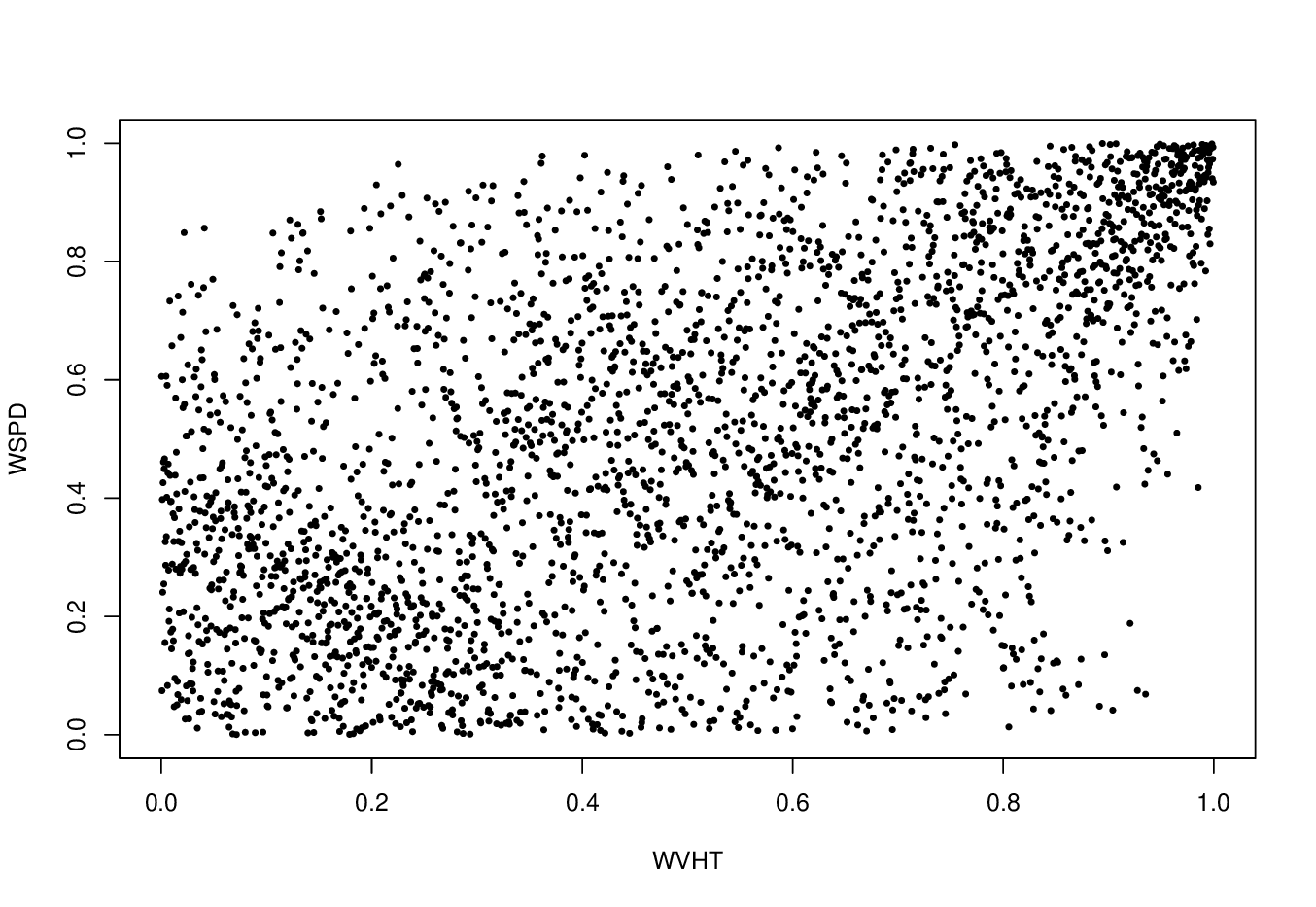}
		\label{fig:speed}
	\end{subfigure}
	\caption{Rank plots of the pairs. \textbf{Left:} WVHT versus APD; \textbf{Right:} WVHT  versus WSPD.} 
	\label{fig:rankplot}
\end{figure}

\begin{table}[H]
	\centering
	\caption{P-values of the tests based on the empirical copula, the empirical Bernstein copula ($m_1=\lfloor n/100\rfloor, m_2=\lfloor n/50\rfloor$) and the empirical beta.}
	\scalebox{0.7}{ 
		\begin{tabular}[t]{lcccccccccc}
			\toprule[1.5pt]
			Pairs&$R_n$&$R_{n,m_1}$ &$R_{n,m_2}$ &$S_n$&$S_{n,m_1}$&$S_{n,m_2}$&$T_n$&$T_{n,m_1}$&$T_{n,m_2}$& $R_n^{\beta}$\\
			\midrule
			(WVHT, APD)&$ 0.000$&     $0.000$&$0.000$ &    $0.002$  &$ 0.000 $&  $0.000$ &  $0.000$&$ 0.000$ &$0.000$&$ 0.000 $ \\
			(WVHT, WSPD) & $ 0.005$&     $0.000$& $0.000$&    $0.022$  &$ 0.045$&$0.005$   &  $0.000$&$ 0.000$ &$0.005$&$ 0.001$ \\
			\bottomrule[1.5pt]
		\end{tabular}
	}
	\label{table:5}
\end{table}

\n To apply the proposed testing procedure, the Bernstein polynomial order was taken to be $m\in \left\{\lfloor 2118/100\rfloor, \right. $\\ $ \left. \lfloor 2118/50\rfloor\right\}$, and other settings (the nominal level, number of multiplier bootstrap replicates $H$ and $20 \times 20 $ grids on $[0, 1]^2$) were the same as in the simulation study. From Table~\ref{table:5}, the Bernstein tests, empirical beta test and empirical tests reach the same conclusion that these two pairs have an asymmetric dependence pattern.

Once the asymmetric characteristic of dependence is confirmed, all the symmetric copula families (which is a huge proportion) should be ruled out for modelling. That is, to construct the copula function of $\rm{(WVHT, APD)}$ and $\rm{(WVHT, WSPD)}$, only asymmetric copulas should be considered. Naturally, one can use the Khoudraji's device in Section~\ref{sec:Sim} to obtain asymmetric copula models. In addition, there are various methods available in the literature, for example the product method in~\cite{Liebscher2008} and the linear convex combination method in~\cite{Wu2014}.

\subsection{Nutrient data application}

A sample of daily intake (in mg) of calcium (Ca), iron (Fe), protein (Pr), vitamin A (vA) and vitamin C (vC) for $n=747$ women aged
between 25 and 50 years has been collected by the United States Department of Agriculture in 1985.
The nutrient data set (see also,~\cite{Genest2012, Bahraoui2018}) is revisited to check the performance of the proposed Bernstein tests. In terms of the suggestion for Bernstein order $m$ in~\cite{Segers2017}, three candidates of $m$ are chosen, $i.e.$, $m=\left\{\lfloor n/30 \rfloor=24, \lfloor n/70 \rfloor=10, \lfloor n/90 \rfloor=8\right\}$, with number of grids $N=20$. 

To make a fair and concise comparison, the results of the characteristic function test $S_{n, \sigma^*}^N$ in~\cite{Bahraoui2018} are also reported here as a competitor. Since P-values therein are computed based on a large amount of bootstraps, it is reasonable to consider a similarly large bootstrap sample, accordingly, $H=5000$ is considered in this  subsection.

Table~\ref{table:6} indicates a surprisingly different conclusion of the Bernstein tests with $m=\lfloor n/70 \rfloor$, that is, only Bernstein tests can reject the pair (Fe, vC) under $5\%$ nominal level. For other pairs, all the tests arrive at the same conclusion. Moreover, the time of calculating the P-values is also recorded based on a desktop PC with processor \textit{11th Gen Intel(R) Core(TM) i5-11400F}  and graphics \textit{NVIDA GeForce RTX 3060} with $15$ repetitions (for time mean and standard deviation). It is important to emphasize that the $S_n$ test is conducted using the \verb+exchTest+ function within the \textbf{R} \verb+copula+ package, which is known for its high optimization. In terms of computational efficiency, the $S_{n, \sigma^*}^N$ test stands out as the most efficient choice.

\begin{table}[H]
	\setlength{\tabcolsep}{9pt}
	\centering
	\caption{P-values and computational performance of empirical tests, empirical beta copula tests and Bernstein tests ($m_1=\lfloor n/30\rfloor, m_2=\lfloor n/70\rfloor, m_3=\lfloor n/90\rfloor$).  Characteristic function test $S^N_{n, \sigma^*}$ in~\citet[]{Bahraoui2018} is also reported here.}
	
	\scalebox{0.65}{ 
		\begin{tabular}[t]{lccccccccccccccccc}
			\toprule[1.5pt]
			\multirow{2}{*}{}&\multirow{2}{*}{Pair}&\multicolumn{14}{c}{Statistics}\\
			\cline{3-16} 
			&&$R_n$&$R_{n,m_1}$&$R_{n,m_2}$&$R_{n,m_3}$ &$S_n$&$S_{n,m_1}$&$S_{n,m_2}$&$S_{n,m_3}$&$T_n$&$T_{n,m_1}$&$T_{n,m_2}$&$T_{n,m_3}$&$R_n^\beta$&$S^N_{n, \sigma^*}$\\

			\midrule
			
			\multirow{10}{*}{P-values}
			&(Ca, Fe)&0.010&0.005&0.007 &0.000&0.003& 0.005&0.008&0.000& 0.002& 0.006&0.007&0.000&0.006& 0.004\\
			&(Ca, Pr)&   0.000     &0.000 &0.000&0.000&0.000&0.000&0.000&     0.000& 0.001&     0.000&0.000&0.000&0.000& $0.000$\\
			&(Ca, vA)& 0.000    &0.000 &0.000  &0.000   &0.000     &0.000& $0.000$&0.000&0.000&0.000&0.000&0.000&0.000&0.000	\\
			&(Ca, vC)&0.143 &0.166 &0.321&0.216&0.158 &0.169 &0.317&0.223&0.065 &0.092&0.276&0.171&0.151& $0.225$\\
			&(Fe, Pr)&0.304 &0.796& 0.956&0.996&0.394 &0.678& 0.954&0.993&0.114 &0.529&0.943&0.998&0.501& $0.895$	\\
			&(Fe, vA)&0.003& 0.004 &0.022&0.005&0.000&0.004 &0.017&0.005&0.005 &0.002& 0.009&0.002&0.001 &$0.008$	\\
			&(Fe, vC)& 0.008& 0.006&0.012&0.004&0.005& 0.005 &0.009&0.004&0.012 &0.012& 0.013&0.004&0.007&$0.007$\\
			&(Pr, vA)&0.007& 0.006&0.029&0.003& 0.001& 0.005& 0.020&0.002&0.024& 0.013&0.029&0.004&0.005& $0.006$\\
			&(Pr, vC)& 0.194& 0.092& $\mathbf{0.036}$&0.143&0.112&0.089&$\mathbf{0.034}$&0.133& 0.070& 0.240&0.052&0.147&0.153&$0.064$\\
			&(vA, vC) &0.549 &0.834& 0.884&0.927&0.602 &0.828& 0.884&0.922&0.180& 0.774&0.913&0.940&0.624&$0.827$ \\
			\midrule
			\multirow{9}{*}{Time}
						&(Ca, Fe)&16.083&159.849&30.219&20.336&24.947& 434.192&81.676&54.421& 22.694&168.899&35.939&25.630& $640.700$&9.014\\
									\multirow{8}{*}{mean}
			&(Ca, Pr)&  15.973    &159.809 &30.560&20.378&24.818&433.909&82.875& 54.546  &22.779&     168.868&35.902&25.256& $830.757$&8.904\\
					\multirow{7}{*}{(secs)}
			&(Ca, vA)& 15.987&159.800&29.610&20.204&24.807 &438.175&79.886 &55.000&25.505&186.549&34.852&25.584&$613.265$&8.276	\\
			&(Ca, vC)&16.710 &163.675&29.626&19.996&26.679 &442.293 &79.889&53.352&23.705 &171.107&34.833&24.712& $631.973$&8.237\\
			&(Fe, Pr)&16.327&160.460&29.553&19.883&25.110&435.617& 79.913&53.364&23.062 &169.563&34.848&24.661& $765.784$&8.238\\
			&(Fe, vA)&16.230& 168.632 &29.685&19.863&31.671&455.719 &79.912&53.367&23.760 &172.449&34.846&24.638&$594.638$&8.183	\\
			&(Fe, vC)& 16.880& 170.615&29.549&19.960&32.841& 444.503&79.911&53.339&22.883 &169.118&34.806&24.664&$570.436$&8.182\\
			&(Pr, vA)&17.143&159.834&29.605&19.933&25.050 &460.171 &79.893 &53.344&28.486&174.660&34.825&24.637& $556.674$&8.309\\
			&(Pr, vC)& 17.843&175.920 & 29.651&20.661&27.116&439.817&79.898&55.361& 23.242&170.081 &34.830&24.671&$565.487$&8.470\\
			&(vA, vC) &16.571&160.764& 29.589&20.425& 27.435&461.440& 79.921&54.761&28.622& 171.081&34.799&24.686&$608.412$ &8.416\\
			\midrule
			\multirow{9}{*}{Time}
						&(Ca, Fe)&0.080&0.206&0.085 &0.118&0.098& 1.216&0.577&0.271& 0.027& 0.129&0.242&0.589& $66.092$&0.223\\
									\multirow{8}{*}{standard}
			&(Ca, Pr)&  0.059& 0.106 &0.244&0.059&0.057&0.120&0.596& 0.295    & 0.063&     0.110&0.237&0.231& $212.529$&0.175\\
				\multirow{7}{*}{deviation}
			&(Ca, vA)& 0.067    &0.087&0.062 &0.317 &0.053     &4.873&0.075 &0.448&1.605&6.344&0.090&0.324&$36.873$&0.324\\
			&(Ca, vC)&0.073& 0.929&0.082&0.073&0.227&2.589 &0.047& 0.091&0.128&1.348&0.070&0.143& $81.312$&0.114\\
			&(Fe, Pr)& 0.067 &0.093& 0.014&0.014&0.087 &0.115&0.102 &0.066&0.038 &0.130&0.091& 0.121& $114.724$&0.121	\\
			&(Fe, vA)&0.121&  7.014&0.064&0.076&2.718&19.865&0.065&0.065&0.053&0.926&0.099 &0.097&$62.509$&0.079	\\
			&(Fe, vC)& 0.238& 7.008&0.013&0.100&0.409& 13.617 &0.123&0.072&0.065&0.154& 0.108&0.141&$7.012$&0.067\\
			&(Pr, vA)&0.708&0.129 &0.053&0.092& 0.101&24.851& 0.068&0.042&0.384&2.948&0.096& 0.100& $5.417$&0.143\\
			&(Pr, vC)& 0.351& 4.343& 0.124&0.177&0.485&12.610&0.062&0.287&0.147 &0.449 &0.080&0.142&$7.623$&0.271\\
			&(vA, vC) &0.277 &0.376&0.047 &0.139&0.157 &27.036&0.128&0.198&0.505&1.337 &0.121&0.097&$82.739$ &0.119\\
			\bottomrule[1.5pt]
		\end{tabular}
	} 
	
	{\raggedright  \scriptsize \hspace{0.5ex} The bold values indicate a different decision for Bernstein tests comparing  with other tests under $5\%$ nominal level. \par}
	
	\label{table:6}
\end{table}

\section{Final remarks \label{sec:conclusion}}

\n Tests based on the Bernstein polynomials for symmetry of bivariate copulas were proposed and investigated. These test statistics are  smoothed versions of those based on the empirical copula in~\cite{Genest2012}. The proposed procedure exhibits enhanced performance in simulation studies and aligns with identical conclusions in real data applications across the majority of scenarios. The limiting distributions of the proposed test statistics were investigated and a Bernstein version of multiplier bootstrap was constructed and implemented to simulate P-values.

Since the underlying copula is continuous, a smooth copula estimator such as the empirical Bernstein copula is competitive with the empirical copula. From the bias-variance trade-off point of view, with appropriate smoothness parameter, the former can outperform the latter by balancing the bias and variance. On the other hand, the smoothed Bernstein tests still hold the same features as the non-smoothed empirical tests. For example, the empirical tests tend to have an empirical level which is below the nominal level, see~\cite{Genest2012} and ~\cite{Bahraoui2018}. The proposed Bernstein tests still undergo this pattern, but are less affected.

For future study, it would be possible to apply the Bernstein polynomials for testing various kinds of symmetry in~\cite{Nelsen1993} and the vertex and diametrical symmetry developed in~\cite{Mangold2017} recently. In general, other statistical tests based on the empirical copula can be adapted easily to use the empirical Bernstein copula. Another avenue to explore involves adopting distinct Bernstein orders for each component of the empirical Bernstein copula. Drawing from the authors' experiential insights, such alternatives are likely to offer benefits, particularly concerning $R_{n,m}$ and $T_{n,m}$.

\section{Acknowledgements}
M. Belalia gratefully acknowledge the research support
of the Natural Sciences and Engineering Research Council of Canada(RGPIN/05496-2020 ).

\appendix

\section{Proof of Theorem~\ref{thm:2}}
\begin{proof}
	This proof is an adaption of~\citet[Proposition 3]{Genest2012}. For the convergence of $nR_{n, m}$ and $n^{1/2}T_{n, m}$, one can directly apply continuous mapping theorem combined with Theorem~\ref{thm:1}. For the convergence of $nS_{n, m}$, the functional delta method was used. To this end, let $\mathscr{C}[0, 1]^2$ denote the space  of continuous functions on $[0, 1]^2$, $\mathscr{D}[0, 1]^2$ denote the space of functions with continuity from upper right quadrant and limits from other quadrants on $[0, 1]^2$, equipped with sup-norm. Further, denote $\rm{BV}_1[0, 1]^2$ by the subspace of $\mathscr{D}[0, 1]^2$ where functions with total variation bounded by $1$. By continuous mapping theorem, 
	\begin{equation*}
		\left(\bS^2_{n, m}, \BB_{n, m}\right)\rightsquigarrow \left(\bS_C^2, \BB_C\right)
	\end{equation*}
	in the space $\ell^{\infty}[0, 1]^2\times \ell^{\infty}[0, 1]^2$. Rewrite it as 
	\begin{equation*}
		\left(\bS^2_{n, m}, \BB_{n, m}\right)=n^{1/2}\{(A_{n, m}, \widehat{C}_{n, m})-(A, C)\},
	\end{equation*}
	where $A\equiv0$ and $A_{n, m}:=n^{1/2}(\widehat{C}_{n, m}-\widehat{C}_{n, m}^{\top})^2$, where $ \widehat{C}_{n, m}^{\top}(u,v) = \widehat{C}_{n, m}(v,u) $. Then, consider the map $\phi: \ell^{\infty}[0, 1]^2\times \rm{BV}_1[0, 1]^2\to \RR$ defined by 
	\begin{equation*}
		\phi(a, b)=\int_{(0, 1]^2}a \dif b,
	\end{equation*}
	Clearly, 
	\begin{equation*}
		nS_{n, m}=\phi	\left(\bS^2_{n, m}, \BB_{n, m}\right)=n^{1/2}\{\phi(A_{n, m}, \widehat{C}_{n, m})-\phi(A, C)\}.
	\end{equation*}
	To conclude the proof, by~\citet[Lemma 4.3]{Carabarin-Aguirre2010}, $\phi$ is Hadamard differentiable tangentially to $\mathscr{C}[0, 1]^2\times \mathscr{D}[0, 1]^2$ at each $(\alpha, \beta)$ in $\ell^{\infty}[0, 1]^2\times \rm{BV}_1[0, 1]^2$ such that $\int |d\alpha|<\infty$ with derivative 
	\begin{equation*}
		\phi'_{(\alpha, \beta)}(a, b)=\int\alpha \dif b+\int a  \dif \beta.
	\end{equation*}
	Then by applying the Functional Delta Method~\citep[Theorem 3.9.4]{Van_der_vaart1996}, $nS_{n, m}\rightsquigarrow\phi'_{(A, C)}\left(\bS_C^2, \BB_C \right)$, where 
	\begin{equation*}
		\phi'_{(A, \,C)}\left(\bS_C^2, \BB_C \right)=\int_{(0, 1]^2}A \dif \BB_C+\int_{(0, 1]^2}\bS_C^2 \dif C=\int_{(0, 1]^2}\bS_C^2 \dif C.
	\end{equation*}
	This yields to the desired result.
\end{proof}
\section{Proof of Theorem~\ref{thm:3}}
\begin{proof}
	This proof is an adaption of~\citet[Proposition 4]{Genest2012}.	The strongly uniform consistency of $\widehat{C}_{n,m}$ was provided in \citet[Theorem 1]{Jansen2012}, then by continuous mapping theorem, it follows immediately that $R_{n,m}$ and $T_{n,m}$ converge to $R_C$ and $T_C$ almost surely, respectively. Further, to prove the convergence of $S_{n,m}$, write
	\begin{equation*}
		|S_{n, m}-S_C|\le |\gamma_{n, m}|+|\zeta_{n}|,
	\end{equation*}
	where 
	\begin{align*}
		\gamma_{n, m}&=\int_{0}^1\int_{0}^1\left\{\widehat{C}_{n, m}(u, v)
		-\widehat{C}_{n, m}(v, u)\right\}^2 \dif \widehat{C}_{n}(u, v)\\
		& \quad-\int_{0}^1\int_{0}^1\left\{{C}(u, v)-{C}(v, u)\right\}^2 \dif \widehat{C}_{n}(u, v),
	\end{align*}
	and 
	\begin{align*}
		\zeta_{n}&=\int_{0}^1\int_{0}^1\left\{{C}(u, v)-{C}(v, u)\right\}^2 \dif \widehat{C}_{n}(u, v)-\int_{0}^1\int_{0}^1\left\{{C}(u, v)-{C}(v, u)\right\}^2 \dif{C}(u, v).
	\end{align*}
	Since 
	\begin{align*}
		&\left|\left\{\widehat{C}_{n ,m}(u, v)-\widehat{C}_{n, m}(v, u)\right\}^2-\left\{C(u, v)-C(v, u)\right\}^2\right|\\
		&=\left|\left[\widehat{C}_{n, m}(u, v)+C(u, v)\right]-\left[\widehat{C}_{n ,m}(v, u)+C(v, u)\right]\right|\\
		&\quad \times\left|\left[\widehat{C}_{n, m}(u, v)-C(u, v)\right]-\left[\widehat{C}_{n, m}(v, u)-C(v, u)\right]\right|\\
		&\le \left[\left|\widehat{C}_{n, m}(u, v)+C(u, v)\right|+\left|\widehat{C}_{n ,m}(v, u)+C(v, u)\right|\right]\\
		&\quad \times\left[\left|\widehat{C}_{n, m}(u, v)-C(u, v)\right|+\left|\widehat{C}_{n, m}(v, u)-C(v, u)\right|\right]\\
		&\le  8\sup_{(u,v)\in [0, 1]^2}\left|\widehat{C}_{n, m}(u, v)-C(u, v)\right|,
	\end{align*}
	one has
	\begin{equation*}
		|\gamma_{n ,m}|\le 8\sup_{(u,v)\in [0, 1]^2}\left|\widehat{C}_{n, m}(u, v)-C(u, v)\right|\xrightarrow{a.s.} 0.
	\end{equation*}
	For $\zeta_{n}$, by~\citet[Proposition A.1 (i)]{Genest1995}, one has
	\begin{equation*}
		\int_{0}^1\int_{0}^1\left\{{C}(u, v)-{C}(v, u)\right\}^2 \dif \widehat{C}_{n}(u, v)\rightarrow\int_{0}^1\int_{0}^1\left\{{C}(u, v)-{C}(v, u)\right\}^2 \dif {C}(u, v),
	\end{equation*}
	then $\zeta_{n} \rightarrow0$. Therefore, $S_{n, m}$ converges to $S_C$ almost surely. 
\end{proof}

\section{Proof of Lemma~\ref{prop:2023-08-24, 4:21PM}}\label{pf:newprop}
\begin{proof}
	Under these assumptions, one need to use the framework in~\cite{Segers2017}. Specifically, let $\mu_{m, (u, v)}$ be the law of random vector $(B_1/m, B_2/m)$, where $B_1$ and $B_2$ follow $\textsf{Binomial}(m, u)$ and $\textsf{Binomial}(m, v)$, respectively. The empirical Bernstein copula can be rewritten as
	\begin{equation*}
		\widehat{C}_{n, m}(u, v)=\int_{[0, 1]^2}\widehat{C}_n(x, y)\,\d\mu_{m, (u, v)}(x, y), \qquad (x, y)\in [0 ,1]^2.
	\end{equation*}
	Moreover, write $(x, y)(t)=(u, v)+t((x, y)-(u, v))$ with $t\in [0, 1]$. Then, the empirical Bernstein copula process is
	\begin{align}\label{eq:2023-08-24, 6:52PM}
		\BB_{n, m}(u, v)&=\sqrt{n}\left\{	\widehat{C}_{n, m}(u, v)-C(u, v)\right\}\notag\\
		&=\sqrt{n}\left\{	\widehat{C}_{n, m}(u, v)-\int_{[0, 1]^2}C(x, y)\,\d\mu_{m, (u, v)}(x, y)+\int_{[0, 1]^2}C(x, y)\,\d\mu_{m, (u, v)}(x, y)-C(u, v)\right\}\notag\\
		&=\int_{[0, 1]^2}\sqrt{n}\left\{	\widehat{C}_{n, m}(x, y)-C(x, y)\right\}\,\d\mu_{m, (u, v)}(x, y)+\sqrt{n}\left\{\int_{[0, 1]^2}C(x, y)\,\d\mu_{m, (u, v)}(x, y)-C(u, v)\right\}\notag\\
		&=T_1+T_2.
	\end{align}
	The two terms are dealt with separately.
	\begin{itemize}
		\item For the term $T_1$, according to~\citet[Proposition 3.1]{Segers2017},  one has 
		\begin{equation*}
			\sup\limits_{(u, v)\in [0, 1]^2}\left|\int_{[0, 1]^2}\sqrt{n}\left\{	\widehat{C}_{n, m}(s, t)-C(s, t)\right\}\,\d\mu_{m, (u, v)}(s, t)-\sqrt{n}\left\{	\widehat{C}_{n, m}(u, v)-C(u, v)\right\}\right|=o_p(1).
		\end{equation*}
		And note that, $\sqrt{n}\left\{	\widehat{C}_{n, m}(u, v)-C(u, v)\right\}\rightsquigarrow\BB_C(u, v)$ in $\ell^{\infty}([0, 1]^2)$ under the Assumption~\ref{ass:1} (see~\cite{Segers2012}). Therefore, $T_1\rightsquigarrow \BB_C(u, v)$ as $n$ goes to infinity.
		
		\item For the term $T_2$, Let $m=cn^{\alpha}$ for some $c >0$, by~\citet[Lemma 3.1]{Kojadinovic2022stute},under Assumption~\ref{ass:1}-\ref{ass:2} , one has 
		\begin{align*}
			&\sup_{(u,v)\in [0, 1]^2}\sqrt{n}\left|\int_{[0, 1]^2}C(x, y)\,\d\mu_{m, (u, v)}(x, y)-C(u, v)\right|=O\left(n^{(3-4\alpha)/6}\right)
		\end{align*}
		almost surely. Therefore,  if $\alpha > 3/4$, $T_2$ goes to zero as $n$ goes to infinity.

	\end{itemize}
	Combining above results completes the proof.
\end{proof}

\section{Proof of Proposition~\ref{prop:1}}
\begin{proof}
	By~\cite{Remillard2009}, one has
	\begin{align*}
		\CC_n(u, v)=\sqrt{n}\left\{C_n(u, v)-C(u, v)\right\}&=n^{1/2}\Bigg\{\frac{1}{n}\sum_{i=1}^n\left\{\II\left(U_i\le u, V_i\le v)-C(u, v\right)\right\}\Bigg\}\\
		& \quad  \rightsquigarrow \CC(u ,v),
	\end{align*}
	where $C_n(u,v)=\frac{1}{n}\sum_{i=1}^{n}\II\left(U_i\le u, V_i\le v\right)$, and
	\begin{align*}
		\overline{\mathbb{C}}^{(h)}_n(u, v)
		&=n^{1/2}\Bigg\{\frac{1}{n}\sum_{i=1}^n\left(\xi^{(h)}_i-\bar{\xi}_n^{(h)}\right)\II(\widehat{U}_i\le u, \widehat{V}_i\le v)  \Bigg\}\rightsquigarrow \CC(u ,v).
	\end{align*}
	To end the proof, one needs to show that the difference between $\left(\CC_n, \overline{\mathbb{C}}^{(h)}_n\right)$ and $\left(\widetilde{\BB}_{n, m}, \overline{\BB}^{(h)}_{n, m}\right)$ are asymptotically negligible.
	
	It is well-known (for example, see~\cite{Deheuvels1979}) that, as $n\to \infty$,
	\begin{equation*}
		\|C_n(u, v)-C(u, v)\|=O\left(n^{-1/2}(\log \log n)^{1/2}\right),
	\end{equation*}
	almost surely and using the same techniques in~\cite{Jansen2012} gives 
	\begin{equation*}
		\|C_{n,m}(u, v)-C(u, v)\|=O\left(n^{-1/2}(\log \log n)^{1/2}\right),
	\end{equation*}
	almost surely, it immediately follows that
	\begin{equation*}
		\|C_{n,m}(u, v)-C_n(u, v)\|= O\left(n^{-1/2}(\log \log n)^{1/2}\right),
	\end{equation*}
	almost surely. Therefore, one can conclude that $\widetilde{\BB}_{n,m}(u, v)\rightsquigarrow \CC(u, v)$. 
	
	Further, by~\citet[Lemma 1]{Jansen2012}, as $n\to \infty$, 
	\begin{align*}
		\|\widehat{C}_{n}(u ,v)-C(u, v)\|&=O\left(n^{-1/2}(\log \log n)^{1/2}\right),
	\end{align*}
	almost surely. By Lemma~\ref{prop:2023-08-24, 4:21PM}, under the assumptions, 
	\begin{equation*}
				\|\widehat{C}_{n,m}(u ,v)-C(u, v)\|=o_p(1).
	\end{equation*}
  Hence
	\begin{equation*}
		\|\widehat{C}_{n,m}(u ,v)-\widehat{C}_n(u ,v)\|= o_p(1),
	\end{equation*}
 one has that $		\overline{\BB}_{n,m}^{(h)}(u, v) \rightsquigarrow \CC(u ,v)$. 
\end{proof}

\section{Proof of Proposition~\ref{prop:2}}
\begin{proof}
	We only show the uniform consistency for 
	\begin{align*}
		\frac{\partial \widehat{C}_{n,m}(u, v)}{\partial u}&=m\sum_{k=0}^{m-1}\sum_{\ell=0}^m\bigg\{\widehat{C}_n\left(\frac{k+1}{m}, \frac{\ell}{m}\right)-\widehat{C}_n\left(\frac{k}{m}, \frac{\ell}{m}\right)\bigg\}\\
		&\quad \cdot P_{m-1, k}(u)P_{m, \ell}(v).
	\end{align*}
	The result for the other partial derivative can be obtained similarly. For any $u\in [b_n, 1-b_n], v\in [0, 1]$, one has
	\begin{align*}
		&\left|\frac{\partial \widehat{C}_{n,m}(u, v)}{\partial u}-\dot{C}_1(u, v)\right|\\
		&\quad\leq\Bigg|m\sum_{k=0}^{m-1}\sum_{\ell=0}^m\bigg\{\widehat{C}_n\left(\frac{k+1}{m}, \frac{\ell}{m}\right)-\widehat{C}_n\left(\frac{k}{m}, \frac{\ell}{m}\right)\\
		&\qquad -{C}\left(\frac{k+1}{m}, \frac{\ell}{m}\right)+{C}\left(\frac{k}{m}, \frac{\ell}{m}\right)\bigg\}P_{m-1, k}(u)P_{m,\ell}(v)\Bigg|\\
		&\qquad  +\Bigg|m\sum_{k=0}^{m-1}\sum_{\ell=0}^m\bigg\{{C}\left(\frac{k+1}{m}, \frac{\ell}{m}\right)-{C}\left(\frac{k}{m}, \frac{\ell}{m}\right)\bigg\}\\
		&\qquad \cdot P_{m-1, k}(u)P_{m, \ell}(v)-\dot{C}_1(u, v)\Bigg|\\
		&\quad =A_1+A_2.
	\end{align*}
	Further, let  $P_{m, k}'(u)=  \frac{1}{u(1-u)}P_{m, k}(u)\left(k -mu\right)$ be the derivative of $P_{m, k}(u)$, then
	\begin{align*}
		A_1&\leq \sum_{k=0}^{m}\sum_{\ell=0}^m\bigg|\widehat{C}_n\left(\frac{k}{m}, \frac{\ell}{m}\right)-{C}\left(\frac{k}{m}, \frac{\ell}{m}\right)\bigg|\\
		&\quad \cdot \left|P_{m, k}'(u)\right| P_{m, \ell}(v)\\
		&\le \sup_{(u,v)\in [0, 1]^2}\bigg|\widehat{C}_n\left(u, v\right)-{C}\left(u, v\right)\bigg|\cdot  \sum_{k=0}^m\left|P_{m, k}'(u)\right|\\
		&=O\left(m^{1/2}n^{-1/2}(\log\log n)^{1/2}\right),
	\end{align*}
	almost surely as $n\to \infty$ and where 
	\begin{equation*}
		\sum_{k=0}^m\left|P_{m, k}'(u)\right|=O(m^{1/2}),
	\end{equation*}
	by~\citet[Lemma 1]{Jansen2014}. 
	
	For dealing with $A_2$, let $\nu_{m, (u, v)}$ be the law of random vector $(S_1/(m-1), S_2/m)$, where $S_1$ and $S_2$ follow $\textsf{Binomial}(m-1, u)$ and $\textsf{Binomial}(m, v)$, respectively.  Therefore, 
	\begin{align*}
		&\sum_{k=0}^{m-1}\sum_{\ell=0}^mC\left(\frac{k+1}{m}, \frac{\ell}{m}\right)P_{m-1, k}(u)P_{m, \ell}(v)\\
		&\quad =\int_{[0, 1]^2} C\left(\left(x+\frac{1}{m-1}\right)\frac{m-1}{m}, y\right)\\
		&\qquad \dif \nu_{m, (u, v)}(x, y),
	\end{align*}
	and 
	\begin{align*}
		&\sum_{k=0}^{m-1}\sum_{\ell=0}^mC\left(\frac{k}{m}, \frac{\ell}{m}\right)P_{m-1, k}(u)P_{m, \ell}(v)\\
		&\quad =\int_{[0, 1]^2} C\left(x\frac{m-1}{m}, y\right)\dif \nu_{m, (u, v)}(x, y).
	\end{align*}
	Using the representation in~\citet[Proof of Proposition 3.4]{Segers2017}, for $0< t <1$, one has 
	\begin{align}\label{eq:2023-08-28, 1:05PM}
		A_2&=
		\Bigg|m\sum_{k=0}^{m-1}\sum_{\ell=0}^m\bigg\{{C}\left(\frac{k+1}{m}, \frac{\ell}{m}\right)-{C}\left(\frac{k}{m}, \frac{\ell}{m}\right)\bigg\}\notag \\
		&\quad \cdot P_{m-1, k}(u)P_{m, \ell}(v)-\dot{C}_1(u, v)\Bigg|\notag\\
		&=\left|\int_0^1\bigg\{\int_{[0, 1]^2} \left[\dot{C}_1\left(\frac{m-1}{m}x+\frac{1+t}{m}, y\right)-\dot{C}_1(u, v)\right] \right.\notag\\
		&\quad\left. \d\nu_{m, (u, v)}(x, y) \bigg\} \d t\right|.
	\end{align}
	Let  $x'(x, t):= \frac{m-1}{m}x+\frac{1+t}{m}$ and $\varepsilon_n= b_n/2$, then one has,
	\begin{align*}
		\eqref{eq:2023-08-28, 1:05PM}&\le\left|\int_0^1\bigg\{\int_{[0, 1]^2} \left[\dot{C}_1\left(x', y\right)-\dot{C}_1(u, v)\right] \right.\\
		&\quad \cdot\II\left(\max(|x'-u|, |y-v|)\le \varepsilon_n\right)\\
		&\quad \left.\dif\nu_{m, (u, v)}(x, y) \bigg\} \dif t\right|\\
		&\quad +\left|\int_0^1\bigg\{\int_{[0, 1]^2} \left[\dot{C}_1\left(x', y\right)-\dot{C}_1(u, v)\right]\right. \\
		&\quad\left. \cdot \II\left(\max(|x'-u|, |y-v|)> \varepsilon_n\right) \dif\nu_{m, (u, v)}(x, y) \bigg\} \d t\right|\\
		&= A_{21}+A_{22}.
	\end{align*}
	Further, the two terms are dealt with separately using the strategy in~\citet[Proof of Lemma 3.1]{Kojadinovic2022stute}.
	\begin{itemize}
		\item For $A_{21}$, under Assumption~\ref{ass:1}-\ref{ass:2}, by~\citet[Lemma 4.3]{Segers2012}, for a constant $L>0$, one has 
		\begin{align*}
			&\left|\dot{C}_1\left(x', y\right)-\dot{C}_1(u, v)\right|\II\left(\max(|x'-u|, |y-v|)\le \varepsilon_n\right)\\
			&\quad \le Lb_n^{-1}\left[|x'-u|+|y-v|\right].
		\end{align*}
		Further, 
		\begin{align}\label{eq:2023-08-28, 2:20PM}
			A_{21}
			&\le Lb_n^{-1}\int_0^1\bigg\{\int_{[0, 1]^2} \left[|x'-u|+|y-v|\right]\notag\\
			&\quad \d\nu_{m, (u, v)}(x, y) \bigg\} \dif t\notag\\
			&\le Lb_n^{-1} \int_{[0, 1]^2} \int_{0}^{1} \left[|x-u|+|y-v|+\left|\frac{x}{m}-\frac{1+t}{m}\right|\right]\notag\\
			&\quad \dif t \dif\nu_{m, (u, v)}(x, y)\notag\\
			&\le Lb_n^{-1} \int_{[0, 1]^2} \int_{0}^{1} \left[|x-u|+|y-v|\right.\notag\\
			&\quad \left.+\frac{x}{m}+\frac{1+t}{m}\right] \dif t \dif\nu_{m, (u, v)}(x, y)\notag\\
			&=Lb_n^{-1} \int_{[0, 1]^2} \left[|x-u|+|y-v|+\frac{x}{m}+\frac{3}{2m}\right]  \dif\nu_{m, (u, v)}(x, y).
		\end{align}
		By Cauchy-Schwarz inequality,
		\begin{align*}
			\eqref{eq:2023-08-28, 2:20PM}&\le Lb_n^{-1}\left[O\left(m^{-1/2}\right)+O(m^{-1})\right]\\
			&=O\left(b_n^{-1}m^{-1/2}\right).
		\end{align*}
		
		\item For $A_{22}$, since $0\le \dot{C}_1\le 1$ and using the result of $A_{21}$, 
		\begin{align*}
			A_{22}	&\le \left|\int_0^1\bigg\{\int_{[0, 1]^2} \frac{2}{\varepsilon_n}\max(|x'-u|, |y-v|)\right. \\
			&\quad\left. \dif\nu_{m, (u, v)}(x, y) \bigg\} \d t\right|\\
			&\le  \left|\int_0^1\bigg\{\int_{[0, 1]^2}\frac{2}{\varepsilon_n}(|x'-u|+|y-v|)\right.\\
			&\quad \left. \dif\nu_{m, (u, v)}(x, y) \bigg\} \d t\right|\\
			&\le \varepsilon_n^{-1}\left[O\left(m^{-1/2}\right)+O(m^{-1})\right]\\
			&=O\left(b_n^{-1}m^{-1/2}\right).
		\end{align*}
	\end{itemize}
	Therefore,
	\begin{align*}
		&\sup_{\substack{v\in [0, 1],\\u\in [b_n, 1-b_n]}}\left|	\frac{\partial \widehat{C}_{n,m}(u, v)}{\partial u}-\dot{C}_1(u, v)\right|\\
		&\quad =O\left(m^{1/2}n^{-1/2}(\log\log n)^{1/2}\right)+O\left(b_n^{-1}m^{-1/2}\right),
	\end{align*}
	almost surely as $n\to \infty$, which completes the proof.
\end{proof}

%
%

\color{black}

\bibliographystyle{chicago}
\bibliography{Mybib}	
\end{document}